\documentclass[11pt]{article}
\usepackage[utf8]{inputenc}
\usepackage[english]{babel}
\usepackage{amsmath,amsfonts,amsthm,graphicx}
\usepackage{natbib}

\DeclareMathAlphabet{\mathpzc}{OT1}{pzc}{m}{it}

\newtheorem{rem}{Remark}[section]

\newcommand{\tr} {\mbox{tr}}

\def\bm{{\mathbf m}}

\def\b1{{\mathbf 1}}

\usepackage[font=small]{caption}

\newtheorem{prop}{Proposition}
\newtheorem{theorem}{Theorem}

\usepackage{titlesec}

\titleformat*{\section}{\normalfont\fontsize{14}{17}\bfseries}
\titleformat*{\subsection}{\normalfont\fontsize{12}{15}\selectfont}

\usepackage{bm}
\usepackage{enumerate}
\providecommand{\abs}[1]{\lvert#1\rvert}
\providecommand{\norm}[1]{\lVert#1\rVert}

\usepackage{soul}
\date{}

\begin{document}

\title{\LARGE Nonparametric multivariate regression estimation for circular responses}\normalsize
\author{Andrea Meil\'an-Vila \\
Universidade da Coru\~{n}a\thanks{%
Research group MODES, CITIC, Department of Mathematics, Faculty of Computer Science, Universidade da Coru\~na, Campus de Elvi\~na s/n, 15071,
A Coru\~na, Spain}
\and %
Mario Francisco-Fern\'andez\\
%EndAName
Universidade da Coru\~{n}a\footnotemark[1]
\and Rosa M. Crujeiras \\
%EndAName
Universidade de Santiago de Compostela\thanks{Department of Statistics, Mathematical Analysis and Optimization, Faculty of Mathematics, Universidade de Santiago de Compostela, R\'ua Lope G\'omez de Marzoa s/n,
	15782, Santiago de Compostela, Spain}
\and
Agnese Panzera\\
%EndAName
Università degli Studi di Firenze\thanks{Dipartimento di Statistica, Informatica, Applicazioni ``G. Parenti'', Università degli Studi di Firenze,  Viale Morgagni, 59, 50134, Firenze, Italy}
}
\maketitle

%.........................................%

\begin{abstract}
	Nonparametric estimators of a regression function with circular response and $\mathbb{R}^d$-valued predictor are considered in this work. Local polynomial type estimators are proposed and studied. Expressions for their asymptotic biases and variances are derived, and some guidelines to select
	asymptotically local optimal bandwidth matrices are also given.
	The finite sample behavior of the proposed estimators is
	assessed through simulations and their performance is also illustrated with a real data set.

\end{abstract}	
\section{Introduction}
\label{intro}

New challenges on regression modeling appear when trying to describe relations between variables and at least some of them do not belong to an Euclidean space. For example, in many situations, one can be interested in estimating regression curves where some or all of the involved variables are \emph{circular} ones. The special nature of circular data (points on the unit circle; angles in $\mathbb{T}=[0,2\pi)$) relies on their \emph{periodicity}, which requires \emph{ad hoc} statistical methods to analyze them. Circular statistics is an evolving discipline, and several statistical techniques for linear data now may claim their circular analogues. Comprehensive reviews on circular statistics (or more general, directional data) are provided in \cite{fisher1995statistical}, \cite{jammalamadaka2001topics} or \cite{mardia2009directional}. Some recent advances in directional statistics are collected in \cite{ley2017modern}. Examples of circular data arise in many scientific fields such as biology, studying animal orientation  \citep{batschelet1981circular}, environmental applications \citep{sengupta2006asymmetric}, or oceanography \citep[as in][among others]{wang2015joint}. In this setting, when the circular variable is supposed to vary with respect to other covariates and the goal is to model such a relation, regression estimators for circular responses must be designed and analyzed.

Parametric approaches were originally considered in \cite{Fisher1992regression} and \cite{presnell1998projected}, assu\-ming a parametric (conditional) distribution model for the circular response. In this scenario, covariates are supposed to influence the response via the parameters of the conditional distribution (e.g. through the location parame\-ter, as the simplest case, or through location and concentration, if a von Mises distribution is chosen). In a practical setting,  in \cite{scapini2002multiple}, the orientation of two species of sand hoppers, considering parametric multiple regre\-ssion methods for circular responses, following the proposal in \cite{presnell1998projected}, is analyzed. A parametric multivariate circular regression problem was also studied in \cite{kim2017multivariate}. Beyond parametric restrictions, flexible approaches are also feasible in this context, just imposing some regularity conditions on the regression function, but avoiding the assumption of a specific parametric family for the regression function or for the conditional distribution. Local estimators of the regression function for circular response and a single real-valued covariate were introduced in \cite{di2013non}. The authors proposed local estimators for the regression function which are defined as the inverse tangent function of the ratio between two sample statistics, obtained as weighted sums of the sines and the cosines of the response, respectively.

In the present work, a regression model with circular response and $\mathbb{R}^d$-valued predictor is considered. When the response variable is circular, the usual target regression function (derived from a cosine risk measure) is given by the inverse tangent function of the ratio between the conditional expectation of the sine and the conditional expectation of the cosine of the response variable.  In this context, nonparametric regression estimators are proposed and studied. Our proposal considers two (separate)  regression models for the sine and cosine components, which are indeed regression models with real-valued res\-ponses. Then, nonparametric estimators for the regression function at hand are obtained by computing the inverse tangent function of the ratio of multivaria\-te  local polynomial estimators for the two  regression functions of the sine and cosine models. This way,  estimators are obtained generalizing, both for higher dimensions and for higher polynomial degrees, the structure of the estimators proposed in \cite{di2013non} for a linear-circular regression function. The approach of considering two flexible regression models for the sine and cosine components has been also explored in \cite{jamma&sarma93}, where the objective is the estimation of a circular-circular regression function. In this case, the conditional expectations of the sine and the cosine of the response  are approximated by trigonometric polynomials of a suita\-ble degree.  A similar approach has been also considered in \cite{di2014non}, where the problem of nonparametrically estimating a spherical-spherical regression is addressed  as a multi-output regression problem. In this case,  each Cartesian coordinate of the spherical regression function is separately estimated, within a scheme of a regression with a linear response and a spherical predictor. A  multivariate angular regression model for both angular and linear predictors was studied by \cite{rivest2016general}. Maximum likelihood estimators for the parameters were derived under two von Mises error structures.

This paper is organized as follows. In Section \ref{sec:model}, the multivariate linear-circular regression model considered in this paper is presented, jointly with the models for the sine and cosine components, establishing certain relations between their first and second order moments. In Section \ref{sec:est}, the nonparametric estimators of the regression function are proposed. Section \ref{sec:NW} and Section \ref{sec:LL} contain the Nadaraya--Watson (NW) and local linear (LL) versions of these estimators, respectively, and include expressions for their asymptotic biases and variances. A local polynomial type estimator with a general degree $p$, for  the univaria\-te case $(d=1)$, is also analyzed in Section \ref{sec:high}. The finite sample performance of the estimators is assessed through a simulation study, provided in Section \ref{sec:sim}. Finally, Section \ref{sec:example} shows a real data application about sand hoppers orientation.

\section{The regression model with circular response}  
\label{sec:model}

Let $\{(\mathbf{X}_i,\Theta_i) \}_{i=1}^{n}$ be a random sample from  $(\mathbf{X},\Theta)$, where $\Theta$ is a circular random variable taking values on $\mathbb{T}=[0,2\pi)$, and $\mathbf{X}$ is a random variable with density $f$ supported on $D\subseteq\mathbb{R}^d$. Assume that $\Theta$ and $\mathbf{X}$ are related through the following regression model:
\begin{equation}\label{model}
\Theta_i=[m(\mathbf{X}_i)+{\varepsilon}_i](\mbox{\texttt{mod}} \, 2\pi), \quad i=1,\dots,n,
\end{equation} 
where $m$ is a \emph{circular regression} function, and the ${\varepsilon}_i$ are  independent and identically distributed (i.i.d.) random angles  (independent of the $X_i$) with zero mean direction and  finite concentration. This implies that ${\rm E}[\sin (\varepsilon)\mid\mathbf{X}=\mathbf{x}]=0$ and $\ell(\mathbf{x}) ={\rm E}[\cos (\varepsilon)\mid\mathbf{X}=\mathbf{x}]< \infty$. Additionally, assume that $\sigma^2_1(\mathbf{x})={\rm Var}[\sin(\varepsilon)\mid\mathbf{X}=\mathbf{x}]<\infty$,
$\sigma^2_2(\mathbf{x})={\rm Var}[\cos(\varepsilon)\mid\mathbf{X}=\mathbf{x}]<\infty$ and
$\sigma_{12}(\mathbf{x})={\rm E}[\sin(\varepsilon)\cos(\varepsilon)\mid\mathbf{X}=\mathbf{x}]<\infty$.
In equation (\ref{model}), \texttt{mod} stands for the modulo operation.

The circular regression function $m$ in model (\ref{model}) is the conditional mean direction of $\Theta$ given $\mathbf X$ which, at a point $\mathbf x$, can be defined as the minimizer of the  risk ${\rm E}\{1-\cos[\Theta-m(\mathbf{X})]\mid \mathbf X=\mathbf{x}\}$,  which is comparable to the $L_2$ risk in the circular setting.  Specifically, the minimizer of this cosine risk  is given by $m(\mathbf{x})=\mbox{atan2}[m_1(\mathbf{x}),m_2(\mathbf{x})]$, where  $m_1(\mathbf{x})={\rm E}[\sin(\Theta)\mid\mathbf{X}=\mathbf{x}]$ and $m_2(\mathbf{x})={\rm E}[\cos(\Theta)\mid\mathbf{X}=\mathbf{x}]$, and the function $\mbox{atan2}(y,x)$ returns the angle between the
$x$-axis and the vector from the origin to $(x, y)$. Then, replacing $m_1$ and $m_2$ by appropriate estimators, an estimator for $m$ can be directly obtained. In particular, a whole class of kernel-type estimators for $m$ at $\mathbf x\in D$ can be defined by considering local polynomial estimators for $m_1(\mathbf x)$ and $m_2(\mathbf x)$. Specifically,  estimators of the form:
\begin{equation}\label{est}
\hat{m}_{\mathbf{H}}(\mathbf{x};p)=\mbox{atan2}[\hat{m}_{1, \mathbf{H}}(\mathbf{x};p),\hat{m}_{2, \mathbf{H}}(\mathbf{x};p)]
\end{equation}
are considered, where for any integer $p\geq 0$, $\hat m_{1, \mathbf{H}}(\mathbf{x};p)$ and $\hat m_{2, \mathbf{H}}(\mathbf{x};p)$ denote  the $p$th order local polynomial estimators (with bandwidth matrix $\mathbf H$) of $ m_1(\mathbf{x})$ and $m_2(\mathbf{x})$, respectively. The special cases $p=0$ and $p=1$ yield a NW (or local constant) type estimator and  a LL type estimator of $m(\mathbf x)$, respectively.

Notice that the proposed approach amounts to consider two (separate)  regression models for the sine and cosine of $\Theta$ on $\mathbf{X}$. In particular,  the
following  regression models for the sine component:
\begin{equation}
\sin(\Theta_i)=m_1(\mathbf{X}_i)+\xi_i\quad i=1,\dots,n,\label{model1}
\end{equation}
and the cosine component: 
\begin{equation}
\cos(\Theta_i)=m_2(\mathbf{X}_i)+\zeta_i\quad i=1,\dots,n,\label{model2}
\end{equation}
are considered, where the $\xi_i$ and the $\zeta_i$  are i.i.d. error terms, satisfying ${\rm E}[\xi\mid\mathbf{X}=\mathbf{x}]={\rm E}[\zeta \mid\mathbf{X}=\mathbf{x}]=0$, $s_1^2(\mathbf{x})={\rm Var}[\xi\mid\mathbf{X}=\mathbf{x}]<\infty$, $s_2^2(\mathbf{x})={\rm Var}[\zeta\mid\mathbf{X}=\mathbf{x}]<\infty$ and $c(\mathbf{x})={\rm E}[\xi\zeta\mid\mathbf{X}=\mathbf{x}]<\infty$  at every $\mathbf{x}\in D$.

Using the sine and cosine addition formulas, it is easy to derive some equations relating certain functions referred to model (\ref{model}), and to models (\ref{model1}) and (\ref{model2}). Specifically, defining
$f_1(\mathbf{x})=\sin [m(\mathbf{x})]$ and $f_2(\mathbf{x})=\cos [m(\mathbf{x})]$, it holds that:
\[
m_1(\mathbf{x})=f_1(\mathbf{x})\ell(\mathbf{x})\quad\mbox{and}\quad
m_2(\mathbf{x})=f_2(\mathbf{x})\ell(\mathbf{x}).
\]
Note that $f_1(\mathbf{x})$ and $f_2(\mathbf{x})$ correspond to the \emph{normalized versions} of $m_1(\mathbf{x})$ and $m_2(\mathbf{x})$, respectively. Indeed, taking into account that  $f_2^2(\mathbf{x})+f_1^2(\mathbf{x})=1$, it can be easily deduced that $\ell(\mathbf{x})=[m^2_1(\mathbf{x})+m_2^2(\mathbf{x})]^{1/2}$. Hence,   $\ell(\mathbf{x})$ amounts to the mean resultant length of $\Theta$ given $\mathbf{X}=\mathbf{x}$, which, taking into account that ${\rm E}[\sin(\varepsilon)\mid \mathbf{X}=\mathbf{x}]=0$ is assumed, also corresponds to the mean resultant length of $\varepsilon$ given $\mathbf{X}=\mathbf{x}$. Additionally, the following explicit expressions for the conditional variances of the error terms involved in models  (\ref{model1}) and (\ref{model2}) can be obtained:
\begin{eqnarray*}
	s_1^2(\mathbf{x})&=&f_1^2(\mathbf{x})\sigma^2_2(\mathbf{x})+2f_1(\mathbf{x})f_2(\mathbf{x})\sigma_{12}(\mathbf{x})+f_2^2(\mathbf{x})\sigma^2_1(\mathbf{x}),\\
	s_2^2(\mathbf{x})&=&f_2^2(\mathbf{x})\sigma^2_2(\mathbf{x})-2f_2(\mathbf{x})f_1(\mathbf{x})\sigma_{12}(\mathbf{x})+f_1^2(\mathbf{x})\sigma^2_1(\mathbf{x}),
\end{eqnarray*}
as well as for the covariance between the error terms in (\ref{model1}) and (\ref{model2}):
\[
c(\mathbf{x})=f_1(\mathbf{x})f_2(\mathbf{x})\sigma^2_2(\mathbf{x})-f_1^2(\mathbf{x})\sigma_{12}(\mathbf{x})+f_2^2(\mathbf{x})\sigma_{12}(\mathbf{x})-f_1(\mathbf{x})f_2(\mathbf{x})\sigma^2_1(\mathbf{x}).
\]

In what follows, $\bm{\nabla}g(\mathbf{x})$ and $\bm{\mathcal{H}}_{g}(\mathbf{x})$ will denote the vector of first-order partial derivatives and  the Hessian matrix of a sufficiently smooth function $g$ at $\mathbf{x}$, res\-pectively. Moreover, for a vector $\mathbf{u}= (u_1,\dots,u_d)^T$ and an integrable function $g$, the multiple
integral $\int\int\dots\int g(\mathbf{u})du_1du_2\dots du_d$ will  be simply denoted as $\int g(\mathbf{u})d\mathbf{u}$. Finally, for any
matrix $\mathbf A$, $\abs{\mathbf{A}}$, $\tr(\mathbf{A})$, $\lambda_{{\max}}(\mathbf{A})$ and $\lambda_{{\min}}(\mathbf{A})$  denote its determinant,
trace, maximum eigenvalue and minimum eigenvalue, respectively.
\section{Properties of kernel-type estimators}
\label{sec:est}
Asymptotic (conditional) bias and variance of the estimator given in (\ref{est}) are derived in this section. 
We will focus on the cases in which $p=0$ and $p=1$. For this, the asymptotic properties of the corresponding NW and LL estimators of $m_j(\mathbf{x})$, $j=1,2$ are firstly recalled. These results are then used to obtain  the asymptotic properties of the estimator presented in (\ref{est}) with polynomial degrees $p=0$ and $p=1$. Finally,  asymptotic properties of local polynomial estimators with  arbitrary order $p$ and $D\subseteq \mathbb{R}$ are also studied.

%----------------------------------------%
\subsection{Nadaraya--Watson type estimator}
\label{sec:NW}

Considering models  (\ref{model1}) and (\ref{model2}), local constant estimators for the regre\-ssion functions $m_j$, $j=1,2$, at a given point $\mathbf{x}\in D\subseteq \mathbb{R}^d$, are respectively defined as:
\begin{equation}
\label{estNW}\hat m_{j, \mathbf{H}}(\mathbf{x};0)=\left\{\begin{array}{lc}\dfrac{\sum_{i=1}^n K_{\mathbf{H}}(\mathbf{X}_i-\mathbf{x})\sin(\Theta_i)}{\sum_{i=1}^n K_{\mathbf{H}}(\mathbf{X}_i-\mathbf{x})}&\text{if $j=1$},\\\\ \dfrac{\sum_{i=1}^n K_{\mathbf{H}}(\mathbf{X}_i-\mathbf{x})\cos(\Theta_i)}{\sum_{i=1}^n K_{\mathbf{H}}(\mathbf{X}_i-\mathbf{x})}&\text{if $j=2$},\end{array}\right.
\end{equation}
where, for $\mathbf{u}\in \mathbb{R}^d$, $K_\mathbf{H}(\mathbf{u})=\abs{\mathbf{H}}^{-1}K(\mathbf{H}^{-1}\mathbf{u})$ is the rescaled version of a $d$-variate kernel function $K$, and $\mathbf{H}$ is a $d\times d$ bandwidth matrix. The resulting estimator $\hat m_{\mathbf H}(\mathbf x;0)$ of $m(\mathbf{x})$, obtained by plugging (\ref{estNW}) in (\ref{est}), corresponds to the multivariate version of the local constant estimator proposed in \cite{di2013non}.

Next, the asymptotic conditional  bias and variance expressions for $\hat{m}_{\mathbf{H}}(\mathbf{x};0)$ are derived. First, using standard theory on the multivariate NW estimator \cite{hardle_muller}, the asymptotic conditional bias and variance of $\hat m_{j, \mathbf{H}}(\mathbf{x};0)$, $j=1,2$, are obtained. This preliminary result is given in Proposition \ref{pro1}. The following assumptions on the  design density, the kernel function and the bandwidth matrix are required.

\begin{enumerate}
	\item [(A1)] The design density $f$ is continuously differentiable at $\mathbf{x}\in D$, and satis\-fies  $f(\mathbf{x})>0$. Moreover,   $s_j^2$ and all second-order derivatives  of the regression functions $m_j$, for $j=1,2$,  are continuous at $\mathbf{x}\in D$, and $s_j^2(\mathbf{x})>0$. 
	\item [(A2)] The kernel $K$ is a spherically symmetric density function, twice conti\-nuously differentiable and with compact support (for simplicity with a nonzero value only if $\norm{\mathbf{u}}\le 1$). Moreover, $\int \mathbf{u}\mathbf{u}^T K(\mathbf{u})d\mathbf{u}=\mu_2(K)\mathbf{I}_d$, where $\mu_2(K)\neq 0$ and $\mathbf{I}_d$ denotes the $d\times d$ identity matrix. It is also assumed that $R(K)=\int K^2(\mathbf{u})d\mathbf u<\infty$. 
	%	\item The bandwidth matrix $\mathbf{H}$ is symmetric and positive definite, with $\mathbf{H}\to 0$ and  $n\abs{\mathbf{H}}\to\infty$, when $n\to\infty$.
	%	
	\item [(A3)]	The bandwidth matrix $\mathbf{H}$ is symmetric and positive definite, with $\mathbf{H}\to 0$ and  $n\abs{\mathbf{H}}\to\infty$, as $n\to\infty$.
	%	The ratio $\lambda_{{\max}}(\mathbf{H})/\lambda_{{\min}}(\mathbf{H})$ is bounded above.
\end{enumerate}

In assumption (A3), $\textbf{H}\to 0$ means that every entry of \textbf{H} goes to $0$. Notice that, since \textbf{H} is symmetric and positive definite, $\textbf{H}\to 0$ is equivalent to $\lambda_{{\max}}(\textbf{H})\to 0$. $\abs{\textbf{H}}$ is a quantity of order $\mathcal{O}\left[\lambda_{{\max}}^d(\textbf{H})\right]$ since $\abs{\textbf{H}}$ is equal to the product of all eigenvalues of $\textbf{H}$.

%.........................................%
\begin{prop}
	\label{pro1}
	Given the random sample $\{(\mathbf{X}_i,\Theta_i)\}_{i=1}^n$  from a density su\-pported on $D\times \mathbb{T}$, assume models $(\ref{model1})$ and $(\ref{model2})$.  Under assumptions $({\rm A}1)$--$({\rm A}3)$, the asymptotic conditional bias  of estimators $\hat m_{j, \mathbf{H}}(\mathbf{x};0)$, for $j=1,2$, at a point $\mathbf{x}$ in the interior of the support of $f$, is:
	\begin{eqnarray}\label{expp0}{\rm E}[\hat m_{j, \mathbf{H}}(\mathbf{x};0)- m_j(\mathbf{x})\mid \mathbf{X}_1,\ldots,\mathbf{X}_n]\nonumber&=&\frac{1}{2}\mu_2(K)\tr(\mathbf{H}^2\bm{\mathcal{H}}_{m_j}(\mathbf{x}))\\&&+  \frac{\mu_2(K)}{f(\mathbf{x})}\bm{\nabla}^T{m_j}(\mathbf{x})\mathbf{H}^2\bm{\nabla}f(\mathbf{x})\nonumber\\&&+  \mathpzc{o}_{\mathbb{P}}[\tr(\mathbf{H}^2)],\end{eqnarray}
	and the conditional variance is:
	\begin{equation}\label{varp0} {\rm Var}[\hat m_{j, \mathbf{H}}(\mathbf{x};0)\mid \mathbf{X}_1,\ldots,\mathbf{X}_n]=\frac{R(K)s_j^2(\mathbf{x})}{n \abs{\mathbf{H}}f(\mathbf{x})}+\mathpzc{o}_{\mathbb{P}}\left(\frac{1}{n \abs{\mathbf{H}}}\right).
	\end{equation}
\end{prop}
%........................................%

Now, using expressions (\ref{expp0}) and (\ref{varp0}), the following theorem provides the asymptotic conditional bias and the asymptotic conditional variance of the estimator $\hat m_{\mathbf{H}}(\mathbf{x};0)$.  Its proof is included in the final Appendix.

\begin{theorem}\label{teoNW}
	Given the random sample $\{(\mathbf{X}_i,\Theta_i)\}_{i=1}^n$  from a density su\-pported on $D\times \mathbb{T}$, assume model $(\ref{model})$. Then, under assumptions $({\rm A}1)$--$({\rm A}3)$, the asympto\-tic conditional  bias of estimator  $\hat{m}_{\mathbf{H}}(\mathbf{x};0)$, at  a fixed interior point $\mathbf{x}$ in the support of $f$, is given by:
	\begin{eqnarray*}
		{\rm E}[\hat{m}_{\mathbf{H}}(\mathbf{x};0)-m(\mathbf{x})\mid\mathbf{X}_1,\dots,\mathbf{X}_n]\nonumber&=&\dfrac{1}{2}\mu_2(K){\rm tr}[\mathbf{H}^2{\bm{\mathcal{H}}}_{m}(\mathbf{x})]\\&&+ \dfrac{\mu_2(K)}{\ell(\mathbf{x})f(\mathbf{x})}{\bm{\nabla}} ^Tm(\mathbf{x})\mathbf{H}^2{\bm{\nabla}}  (\ell f)(\mathbf{x})\nonumber\\&&+ {o}_{\mathbb{P}}[{\rm tr}(\mathbf{H}^2)],
	\end{eqnarray*}
	and the asymptotic conditional variance is:
	$$
	{\rm Var}[\hat{m}_{\mathbf{H}}(\mathbf{x};0)\mid\mathbf{X}_1,\dots,\mathbf{X}_n]=\dfrac{R(K)\sigma^2_1(\mathbf{x})}{n\abs{\mathbf{H}}\ell^2(\mathbf{x})f(\mathbf{x})}+\mathpzc{o}_{\mathbb{P}}\left(\dfrac{1}{n\abs{\mathbf{H}}}\right).
	$$
\end{theorem}

\begin{rem}	Note that both the asymptotic conditional  bias and the asymptotic conditional variance share the form of the corresponding quantities for the NW estimator of a regression function with real-valued response. 
	In the asymptotic bias expression, both the gradient and the Hessian matrix of $m$ refer to a circular regression function. In addition, the asymptotic conditional variance depends on the ratio $\sigma^2_1(\mathbf{x})/\ell^2(\mathbf{x})$, accounting for the variability of the errors in model $(\ref{model})$. 
\end{rem}

From Theorem \ref{teoNW}, it is possible to define the asymptotic (conditional) mean squared error ($\rm AMSE$) of 
$\hat{m}_{\mathbf{H}}(\mathbf{x};0)$, as the sum of the square of the main term of bias and the main term of the variance,
\begin{eqnarray}\label{AMSENW}
&&\mbox{AMSE}[\hat{m}_{\mathbf{H}}(\mathbf{x};0)]\nonumber\\&=&\Bigg\{\dfrac{1}{2}{\mu_2(K)}{\rm tr}[\mathbf{H}^2{\bm{\mathcal{H}}}_{m}(\mathbf{x})]+\dfrac{\mu_2(K)}{\ell(\mathbf{x})f(\mathbf{x})}{\bm{\nabla}} ^Tm(\mathbf{x})\mathbf{H}^2{\bm{\nabla}}  (\ell f)(\mathbf{x})\Bigg\}^2+\dfrac{R(K)\sigma^2_1(\mathbf{x})}{n\abs{\mathbf{H}}\ell^2(\mathbf{x})f(\mathbf{x})}\nonumber\\\nonumber&=&\dfrac{1}{4}\mu^2_2(K){\rm tr}^2\Bigg(\mathbf{H}^2\left\{\dfrac{1}{\ell(\mathbf{x})f(\mathbf{x})}[{\bm{\nabla}} (\ell f)(\mathbf{x}){\bm{\nabla}} ^T m(\mathbf{x})+{\bm{\nabla}} m(\mathbf{x}){\bm{\nabla}} ^T (\ell f)(\mathbf{x})]+{\bm{\mathcal{H}}}_{m}(\mathbf{x})\right\}\Bigg)\\&&+ \dfrac{R(K)\sigma^2_1(\mathbf{x})}{n\abs{\mathbf{H}}\ell^2(\mathbf{x})f(\mathbf{x})}.
\end{eqnarray}

The minimizer of equation (\ref{AMSENW}), with respect to $\mathbf H$, provides an asymptoti\-cally optimal local bandwidth matrix for $\hat{m}_{\mathbf{H}}(\mathbf{x};0)$, which is given by:
\begin{eqnarray}
\mathbf{H}_{ \text{opt}}(\mathbf{x})&=&h^*(\mathbf{x})\left[\tilde{\mathcal{B}}(\mathbf x)\right]^{-1/2}\nonumber\\&=&\left[\dfrac{R(K)\sigma^2_1(\mathbf{x})}{nd\mu^2_2(K)f(\mathbf{x})}\abs{\tilde{\mathcal{B}}(\mathbf{x})}^{1/2}\right]^{1/{d+4}}\cdot\left[\tilde{\mathcal{B}}(\mathbf{x})\right]^{-1/2},
\label{ob1NW}
\end{eqnarray}
where
$$\tilde{\mathcal{B}}(\mathbf{x})= \left\{ \begin{array}{lcc}
\mathcal{B}(\mathbf{x}) &   \text{ if }   & \mathcal{B}(\mathbf{x}) \text{ is positive definite,} \\
-\mathcal{B}(\mathbf{x}) & \text{if} & \mathcal{B}(\mathbf{x}) \text{ is negative definite,} \\
\end{array}\right.$$
with 
$$\mathcal{B}(\mathbf{x})=\dfrac{1}{\ell(\mathbf{x})f(\mathbf{x})}[{\bm{\nabla}} (\ell f)(\mathbf{x}){\bm{\nabla}} ^T m(\mathbf{x})+{\bm{\nabla}} m(\mathbf{x}){\bm{\nabla}} ^T (\ell f)(\mathbf{x})]+{\bm{\mathcal{H}}}_{m}(\mathbf{x}).$$

%...................................%

This optimization result can be proved using Proposition 2.6 included in \cite{liu2001kernel}.
Note that in the expression of $\mathbf{H}_{ \text{opt}}(\mathbf{x})$, the matrix  $\tilde{\mathcal{B}}(\mathbf{x})$  determines the shape and the orientation in the $d$-dimensional space of the covariate region which is used to locally compute the estimator. Such data regions for computing the estimator are ellipsoids in $\mathbb R^d$, being the magnitude of the axes controlled by $\tilde{\mathcal{B}}(\mathbf{x})$ . In the particular case of $\bm H=h\mathbf{I}_d$,  the estimator $\hat m_{\mathbf{H}}(\mathbf{x};0)$, with $\mathbf{x}$ being an interior point of the support, achieves an optimal convergence rate of $n^{-4/(d+4)}$, which is the same as the one for the multivariate NW estimator with real-valued response.

Despite deriving the previous explicit expression for the local optimal bandwidth (\ref{ob1NW}), its use in practice is limited given that it depends on unknown functions, such as the design density $f$ and the variance of the sine of the errors $\sigma_1^2$. In addition, when the goal is to reconstruct the whole regression function and the focus is not only set on a specific point, it is more usual in practice to consider a global bandwidth for estimation rather than pursuing an estimator based on local bandwidths. An asymptotic global optimal bandwidth matrix $\mathbf{H}$ could be obtained by minimizing a global error measurement (such as the integrated version of the AMSE). Again, this will depend on unknowns and, moreover, this optimization problem is not trivial, not being possible to obtain a closed form solution. Alternatively, a cross-validation criterion suitably adapted for this context can be used to select the bandwidth matrix. This is indeed the bandwidth selection method employed in our numerical analysis and our real data application. More details will be provided in Section \ref{sec:sim}.

%......................................................%
\subsection{Local linear type estimator}
\label{sec:LL}
Similarly to the case when $p=0$, the local linear case, corresponding to $p=1$, is considered. Specifically, for models ($\ref{model1}$) and ($\ref{model2}$),  the  LL estimators of the regression functions $m_j$, $j=1,2$, at $\mathbf{x}\in D$, are defined by:
\begin{equation}\label{estLL}
\hat m_{j, \mathbf{H}}(\mathbf{x};1)=\left\{\begin{array}{lc}\mathbf{e}_1^T(\bm{\mathcal{X}}_{\mathbf{x}}^T\bm{\mathcal{W}}_{\mathbf{x}}\mathcal{X})^{-1}\bm{\mathcal{X}}_{\mathbf{x}}^T\bm{\mathcal{W}}_{\mathbf{x}}\mathbf{\bm{\mathcal{S}}}&\text{if $j=1$},\\\\ \mathbf{e}_1^T(\bm{\mathcal{X}}_{\mathbf{x}}^T\bm{\mathcal{W}}_{\mathbf{x}}\bm{\mathcal{X}}_{\mathbf{x}})^{-1}\bm{\mathcal{X}}_{\mathbf{x}}^T\bm{\mathcal{W}}_{\mathbf{x}}\mathbf{\bm{\mathcal{C}}}&\text{if $j=2$},\end{array}\right.
\end{equation}
where $\mathbf{e}_1$ is a $(d + 1) \times 1$ vector having 1 in the first entry and 0 in all other entries, $\bm{\mathcal{X}}_{\mathbf{x}}$ is a $n\times(d+1)$ matrix having $(1, (\mathbf{X}_i-\mathbf{x})^T)$ as its  $i$th row, $\bm{\mathcal{W}}_{\mathbf{x}}=\mbox{diag}\{K_\mathbf{H}(\mathbf{X}_1-\mathbf{x}),\dots,K_\mathbf{H}(\mathbf{X}_n-\mathbf{x})\}$, $\mathbf{\bm{\mathcal{S}}}=(\sin(\Theta_1),\dots,\sin(\Theta_n))^T$ and $\mathbf{\bm{\mathcal{C}}}=(\cos(\Theta_1),\dots,\cos(\Theta_n))^T$. 

Using known asymptotic results for the multivariate local linear estimator \citep{ruppert1994multivariate}, the asymptotic conditional bias and variance of $\hat m_{j, \mathbf{H}}(\mathbf{x};1)$, $j=1,2$, can be obtained. These expressions are provided in the following result.

%....................................%
\begin{prop}
	\label{pro2}
	Given the random sample $\{(\mathbf{X}_i,\Theta_i)\}_{i=1}^n$  from a density su\-pported on $D\times \mathbb{T}$,  assume models $(\ref{model1})$ and $(\ref{model2})$.  Under assumptions $({\rm A}1)$--$({\rm A}3)$, asymptotic conditional bias of estimators $\hat m_{j, \mathbf{H}}(\mathbf{x};1)$, $j=1,2$, with $\mathbf x$ being a point in the interior of the support of $f$, is: 
	\begin{eqnarray}\label{expp1}{\rm E}[\hat m_{j, \mathbf{H}}(\mathbf{x};1)-m_j(\mathbf{x})\mid \mathbf{X}_1,\ldots,\mathbf{X}_n]&=&\frac{1}{2}\mu_2(K) \tr(\mathbf{H}^2 \bm{\mathcal{H}}_{m_j}(\mathbf{x}))\nonumber\\&&+  \mathpzc{o}_{\mathbb{P}}[\tr(\mathbf{H}^2)],\end{eqnarray}
	and the asymptotic conditional variance is:
	\begin{equation}\label{varp1} {\rm Var}[\hat m_{j, \mathbf{H}}(\mathbf{x};1)\mid \mathbf{X}_1,\ldots,\mathbf{X}_n]=\frac{R(K)s_j^2(\mathbf{x})}{n \abs{\mathbf{H}}f(\mathbf{x})}+\mathpzc{o}_{\mathbb{P}}\left(\frac{1}{n \abs{\mathbf{H}}}\right).
	\end{equation}
\end{prop}

The resulting estimator $\hat m_{\mathbf H}(\mathbf x;1)$ of $m(\mathbf{x})$ given in  (\ref{est}) corresponds to the multivariate version of the local linear estimator proposed in \cite{di2013non}. The follo\-wing theorem provides the asymptotic conditional bias and the asymptotic conditional variance of this estimator.  Its proof is included in the final Appendix.

\begin{theorem}\label{teoLL} 
	Given the random sample $\{(\mathbf{X}_i,\Theta_i)\}_{i=1}^n$  from a density su\-pported on $D\times \mathbb{T}$,  assume model $(\ref{model})$. Then, under assumptions $({\rm A}1)$--$({\rm A}3)$, the asympto\-tic conditional bias of estimator  $\hat{m}_{\mathbf{H}}(\mathbf{x};1)$, with $\mathbf{x}$ being a fixed interior point in the support of $f$,  is given by:
	$$
	{\rm E}[\hat{m}_{\mathbf{H}}(\mathbf{x};1)-m(\mathbf{x})\mid\mathbf{X}_1,\dots,\mathbf{X}_n]=\dfrac{1}{2}\mu_2(K){\rm tr}[\mathbf{H}^2{\bm{\mathcal{H}}}_{m}(\mathbf{x})]+\dfrac{\mu_2(K)}{\ell(\mathbf{x})}{\bm{\nabla}}^T m(\mathbf{x})\mathbf{H}^2{\bm{\nabla}} \ell(\mathbf{x})+\mathpzc{o}_{\mathbb{P}}[\tr(\mathbf{H}^2)],
	$$
	while its  asymptotic conditional variance is:
	$$
	{\rm Var}[\hat{m}_{\mathbf{H}}(\mathbf{x};1)\mid\mathbf{X}_1,\dots,\mathbf{X}_n]=\dfrac{R(K)\sigma^2_1(\mathbf{x})}{n\abs{\mathbf{H}}\ell^2(\mathbf{x})f(\mathbf{x})}+\mathpzc{o}_{\mathbb{P}}\left(\dfrac{1}{n\abs{\mathbf{H}}}\right).
	$$
\end{theorem}
\begin{rem}
	Estimators $\hat m_{\mathbf H}(\mathbf{x};0)$ and $\hat m_{\mathbf H}(\mathbf{x};1)$ have the same leading terms in their asymptotic conditional variances, while their asymptotic conditional biases, also being of the same order,  have different leading terms. In particular, the main term of the  asymptotic conditional bias of $\hat m_{\mathbf{H}}(\mathbf{x};1)$ does not depend  on the design density, $f$. Moreover,  as a consequence of its definition, the LL type estimator, differently from the NW type one, automatically adapts to boundary regions, in the sense that for compactly supported $f$, the asymptotic conditional bias has the same order both for the interior and for the boundary of the support of $f$  \citep{ruppert1994multivariate}.
\end{rem}

\begin{rem} For $d=1$, asymptotic results  for  estimators having the same form as the univariate version of estimator (\ref{est})  with $p=0$ and $p=1$,  are provided in \cite{di2013non}. Despite they used slightly different formulations for their nonparametric estimators, their results, at interior points, can be directly compared with those obtained in Theorems \ref{teoNW} and \ref{teoLL}. This correspondence is immediately clear for the asymptotic bias terms. For the asymptotic variance, the equivalence between the expressions can be obtained considering the relations between the variance of the error term in model (\ref{model}) with the variance of the error terms in models (\ref{model1}) and (\ref{model2}):
\begin{eqnarray}\label{eq:relation_DM}
f_1^2(\mathbf{x})[m_2^2(x)+s_2^2(x)]+f_2^2(\mathbf{x})[m_1^2(x)+s_1^2(x)]-2f_1(\mathbf{x})f_2(\mathbf{x})[m_1(\textbf{x})m_2(\textbf{x})+{c}(\mathbf{x})]=\sigma^2_1(x)\end{eqnarray}
\end{rem}

As a consequence of Theorem \ref{teoLL}, and similarly to the NW case, an asymptotically optimal local bandwidth can be also obtained for $\hat m_{\mathbf H}(\mathbf{x};1)$, which coincides with (\ref{ob1NW}), but taking $ \mathcal{B}(\mathbf{x})=\ell^{-1}(\mathbf{x})[{\bm{\nabla}} \ell (\mathbf{x}){\bm{\nabla}}^T m(\mathbf{x})+{\bm{\nabla}} m(\mathbf{x}){\bm{\nabla}}^T \ell (\mathbf{x})]+{\bm{\mathcal{H}}}_{m}(\mathbf{x})$.

%............................................%
\subsection{Higher order polynomials}
\label{sec:high}
Standard local polynomial theory \citep{fan1996local} can be used to generalize the above results to local polynomial estimators of arbitrary order $p$. Using similar arguments to those used to prove Theorems 1\ref{teoNW} and \ref{teoLL}, it can be derived that the The conditional bias of the $p$th order polynomial type estimator given in (\ref{est}) will be of order $\mathcal{\mathcal{O}}_{\mathbb{P}}\{[{\rm tr}(\mathbf{H}^2)]^{(p+1)/2}\}$. Moreover, if $p$ is even, $f$ has a continuous derivative in a neighborhood of $\mathbf{x}$, and $\mathbf{x}$ is an interior point of the support of the design density $f$, then the bias will be of order $\mathcal{\mathcal{O}}_{\mathbb{P}}\{[{\rm tr}(\mathbf{H}^2)^{p/2+1}]\}$. 
Here, following the lines in \citet{ruppert1994multivariate}, we will only focus on the case $d=1$ to analyze asymptotically the nonparametric regression estimator given in (\ref{est}) for $p>1$. In particular,  the $p$th degree local polynomial estimators for $m_j$, $j=1,2$, at $x\in \mathcal{D}\subseteq \mathbb{R}$, are:
\begin{equation}\label{locpoly}
\hat m_{j, h}(x;p)=\left\{\begin{array}{lc}\mathbf{e}^T_1(\bm{\mathcal{X}}_{{x},p}^T\bm{\mathcal{W}}_{x}\bm{\mathcal{X}}_{{x},p})^{-1}\bm{\mathcal{X}}_{{x},p}^T\bm{\mathcal{W}}_{x}\bm{\mathcal{S}}&\text{if  $j=1$}\\\\\mathbf{e}^T_1(\bm{\mathcal{X}}_{{x},p}^T\bm{\mathcal{W}}_{x}\bm{\mathcal{X}}_{{x},p})^{-1}\bm{\mathcal{X}}_{{x},p}^T\bm{\mathcal{W}}_{x}\bm{\mathcal{C}}&\text{if $j=2$}\end{array}\right.
\end{equation}
where $\mathbf{e}_1$ is a $(p+1) \times 1$ vector having 1 in the first entry and zero elsewhere, $\bm{\mathcal{X}}_{{x},p}$ is for $n\times p$ matrix with the $(i,k)$th entry equal to $(X_i-x)^{k-1}$, and $\bm{\mathcal{W}}_{x}$ is a diagonal matrix of order $n$ with $(i,i)$th entry equal to $K_h(X_i-x)$, where $K_h(u)=1/h K(u/h)$, being $K$ a univariate kernel function, and $h$ the bandwidth or smoothing parameter. In this univariate framework, the $p$th degree local polynomial type estimator of $m$ at $x$, denoted by $\hat m_{h}(x;p)$, has the same expression as the one given in (\ref{est}), but using estimators $\hat m_{j, h}(x;p)$, $j=1,2$, defined in $(\ref{locpoly})$, as the arguments of the atan2 function.

Let $K_{(p)}$ be the equivalent kernel function defined in \cite{lejeune1992smooth}, which is a kernel of order $p+2$ when $p$ is even and of order $p+1$ otherwise. Let  $\mu_j(K_{(p)})$ and $R(K_{(p)})$ denote the moment of order $j$ and the roughness of $K_{(p)}$,  respectively. Under suitable adaptations of assumptions $({\rm A}1)$--$({\rm A}3)$ to the univariate case and using asymptotic results for standard local polynomial estimators of an arbitrary order $p$, the asymptotic conditional bias and variance of $\hat m_{j,h}(x;p)$, $j=1,2$, can be obtained. It is clear that the conditional asymptotic bias of $\hat m_{h}(x;p)$ will depend on whether the polynomial degree is even or odd.  Since computations are tedious for high-order polynomials, asymptotic  properties of estimator 
$\hat m_{h}(x;p)$ at $x\in \mathcal{D}$ will be derived only when the polynomial degree $p$ is equal to two and three.  Notice that in the case of the regression function, \citet{fan1996local} recommend to use  polynomial orders   $p = 1$ or $p=3$ for estimating this curve. Results could be extended for higher-order polynomial degrees.

\begin{theorem}\label{C_teo2} 
	Let $\{(X_i,\Theta_i)\}_{i=1}^{n}$ be a random sample from a density defined on $\mathcal{D}\times\mathbb{T}$, with $\mathcal{D}\subseteq \mathbb{R}$, and let $x$ be an interior point of the support of the design density $f$. Under assumptions \textnormal{(A1)--(A3)} with $d=1$, and assuming that $m_j$, $j=1,2$, admits continuous derivatives up to order four in a neighborhood of $x$, then
	\begin{eqnarray*}{\mathbb{E}}[\hat m_{h}(x;2)-m(x)\mid X_1,\ldots, X_n]&=&\frac{h^{4}\mu_{4}(K_{(2)})f'(x)}{3!f(x)}[m^{(3)}(x)+a(x)]\\&&+ \frac{h^{4}\mu_{4}(K_{(2)})}{4!}[m^{(4)}(x)+b(x)]+{o}_{\mathbb{P}}\left(h^{4}\right),\end{eqnarray*}
	and
	\[{\mathbb{V}{\rm ar}}[\hat m_{h}(x;2)\mid X_1,\ldots, X_n]=\frac{R\left(K_{(2)}\right)}{n h\ell^2(x) f(x)}\sigma_1^2(x)+{o}_{\mathbb{P}}\left(\frac{1}{n h}\right),\]
		where
	$$a(x)=\dfrac{2\ell''(x)m'(x)+4\ell'(x)m''(x)}{\ell(x)}+\dfrac{m_2''(x)m_1'(x)-m_1''(x)m_2'(x)+2\ell'^2(x)m'(x)}{\ell^2(x)}$$
	and
	\begin{eqnarray*}b(x)&=&\dfrac{2\ell^{(3)}(x)m'(x)+6\ell'(x)m^{(3)}(x)+6\ell''(x)m''(x)}{\ell(x)}\\&&+ \dfrac{2m_2^{(3)}(x)m_1'(x)-2m_1^{(3)}(x)m_2'(x)+6\ell'^2(x)m''(x)+6\ell'(x)\ell''(x)m'(x)}{\ell^2(x)}\end{eqnarray*}
\end{theorem}

\begin{theorem}\label{C_teo3} 
	Let $\{(X_i,\Theta_i)\}_{i=1}^{n}$ be a random sample from a density defined on $\mathcal{D}\times\mathbb{T}$, with $\mathcal{D}\subseteq \mathbb{R}$, and let $x$ be an interior point of the support of the design density $f$. Under assumptions \textnormal{(A1)--(A3)} with $d=1$, and assuming that $m_j$, $j=1,2$, admits continuous derivatives up to order five in a neighborhood of $x$, then
	\begin{eqnarray*}{\mathbb{E}}[\hat m_{h}(x;3)-m(x)\mid X_1,\ldots, X_n]&=&\frac{h^{4}\mu_{4}(K_{(3)})}{4!}[m^{(4)}(x)+b(x)]+{o}_{\mathbb{P}}\left(h^{4}\right),\end{eqnarray*}
	and
	\[{\mathbb{V}{\rm ar}}[\hat m_{h}(x;3)\mid X_1,\ldots, X_n]=\frac{R\left(K_{(3)}\right)}{n h\ell^2(x) f(x)}\sigma_1^2(x)+{o}_{\mathbb{P}}\left(\frac{1}{n h}\right).\]
\end{theorem}

\section{Simulation study}
\label{sec:sim}
In order to illustrate the performance of the estimators proposed in Section \ref{sec:est},  a simulation study considering different scenarios is carried out for $d=2$ (that is, considering a circular response and a bidimensional covariate). For each scenario, 500 samples of size $n$ ($n=64, 100, 225$ and $400$) are generated on a bidimensional regular grid in the 
unit square considering the following regression models:
\begin{enumerate}[{M}1.]
	\item$\Theta=[\mbox{atan2}(6{X}_{1}^5-2{X}_{1}^3-1,-2{X}_{2}^5-3{X}_{2}-1)+\varepsilon](\mbox{\texttt{mod}} \, 2\pi)$,
	\item $\Theta=[\mbox{acos}({X}_{1}^5-1)+\dfrac{3}{2}\mbox{asin}({X}_{2}^3-{X}_{2}+1)+\varepsilon](\mbox{\texttt{mod}}\, 2\pi)$,
\end{enumerate}
where $\mathbf{X}=(X_{1},X_{2})$ denotes the bidimensional covariate, and the circular errors, $\varepsilon$, are drawn from a von Mises distribution $vM(0,\kappa)$ with different values of $\kappa$ (5, 10 and 15).

Figure \ref{figure:sim} shows two realizations of simulated data (model M1: top row; model M2: bottom row). In both cases, the sample size is $n=225$. Left plots show the regression functions evaluated in the regularly spaced sample $(X_1,X_2)$. Central panels present the random errors generated from a von Mises distribution with zero mean direction and concentration $\kappa=5$, for model M1, and $\kappa=15$, for model M2. Right panels show the values of the response variables, obtained \emph{adding}  regression functions and  circular errors.  It can be seen that the errors in the top row, corresponding to $\kappa=5$,  present more variability than the ones generated with $\kappa=15$.

\begin{figure}[t]
	\centering
	\includegraphics[width=0.32\textwidth]{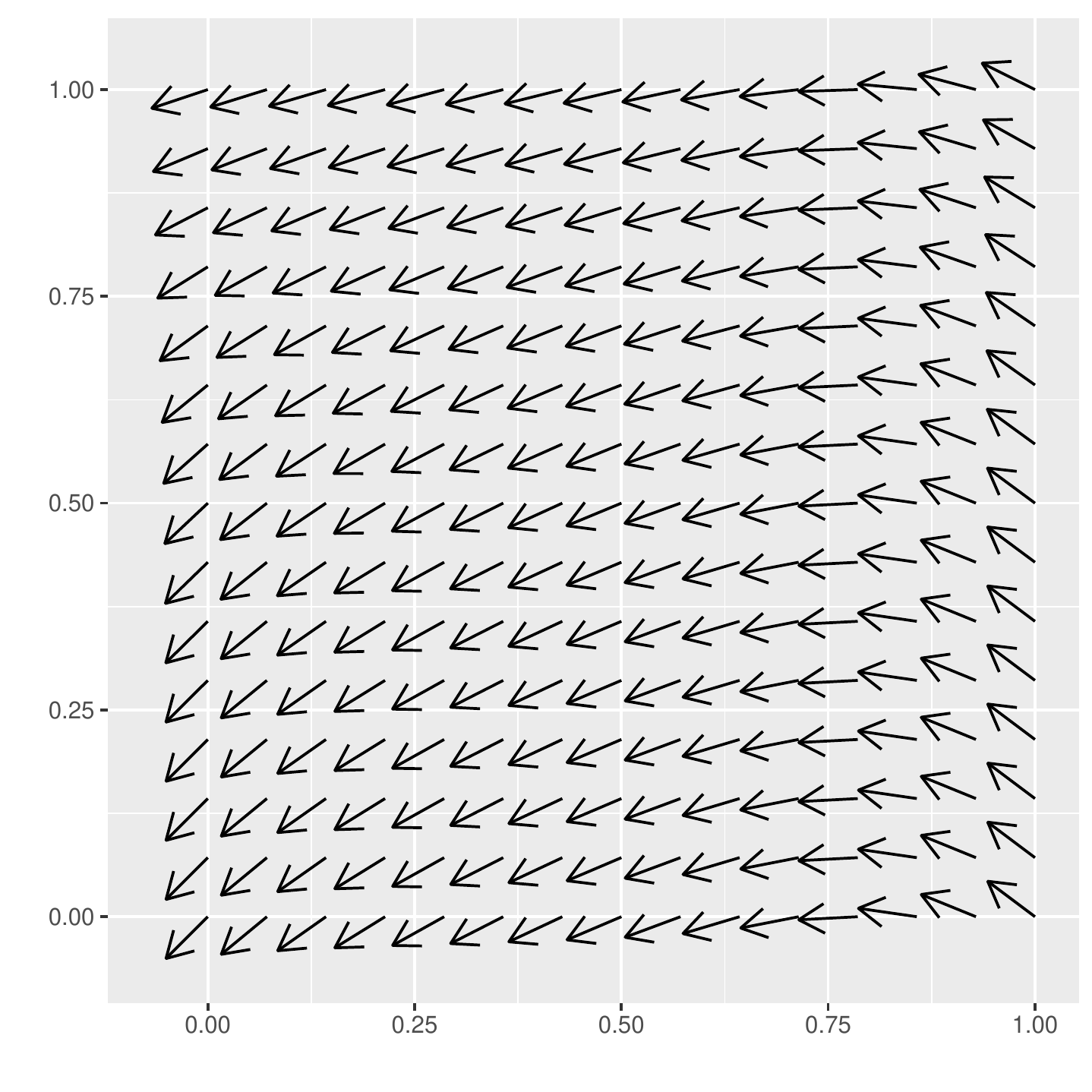}
	\includegraphics[width=0.32\textwidth]{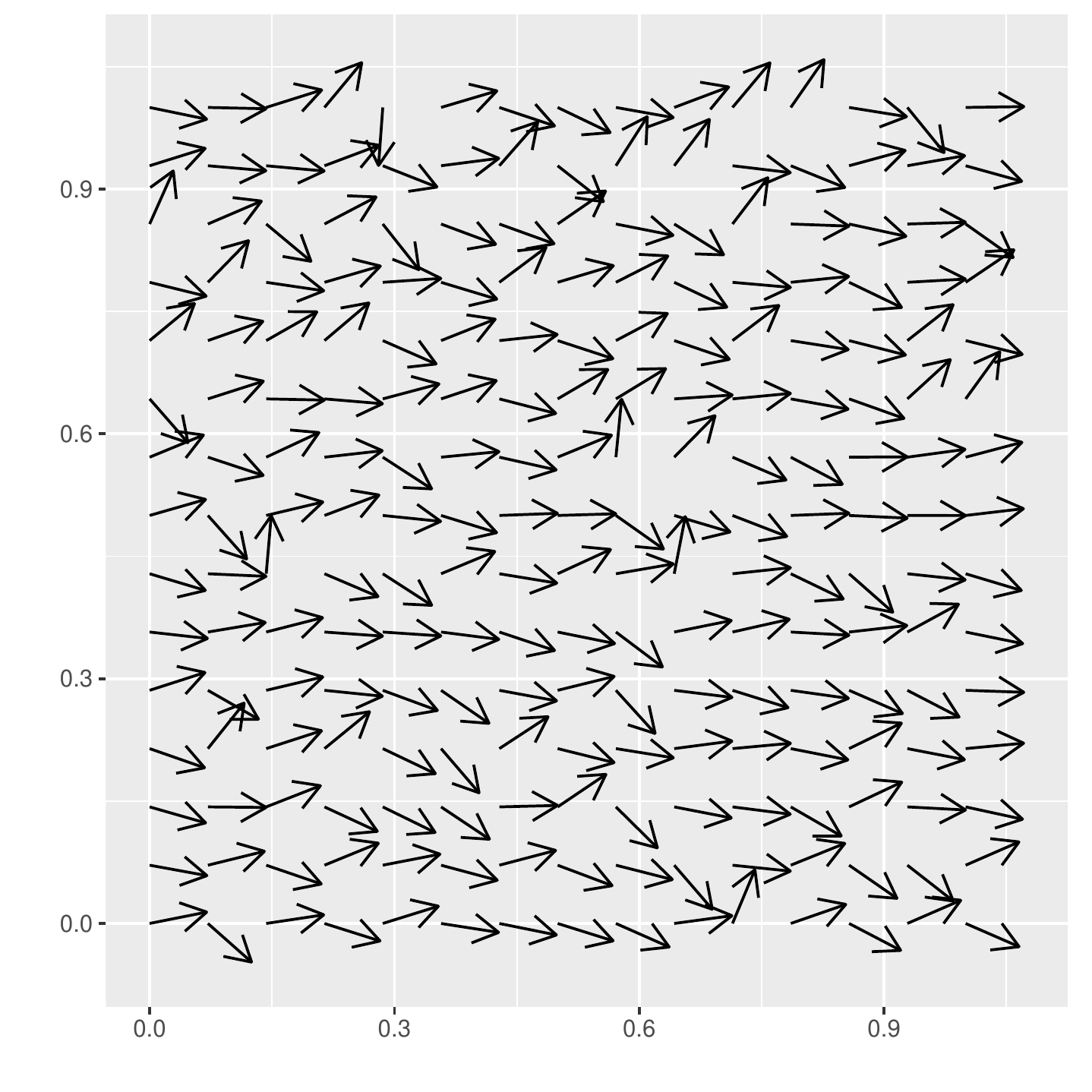}
	\includegraphics[width=0.32\textwidth]{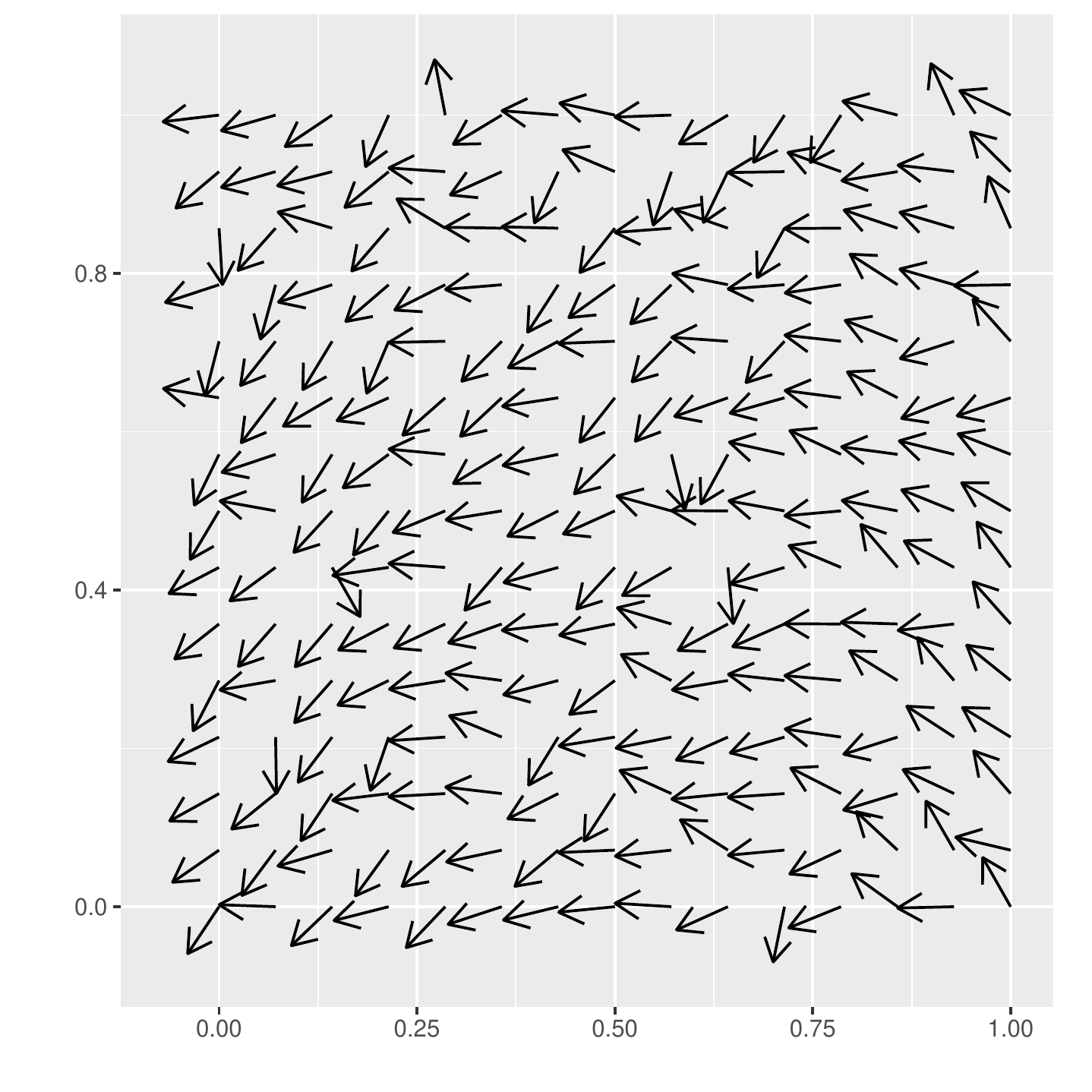}
	\includegraphics[width=0.32\textwidth]{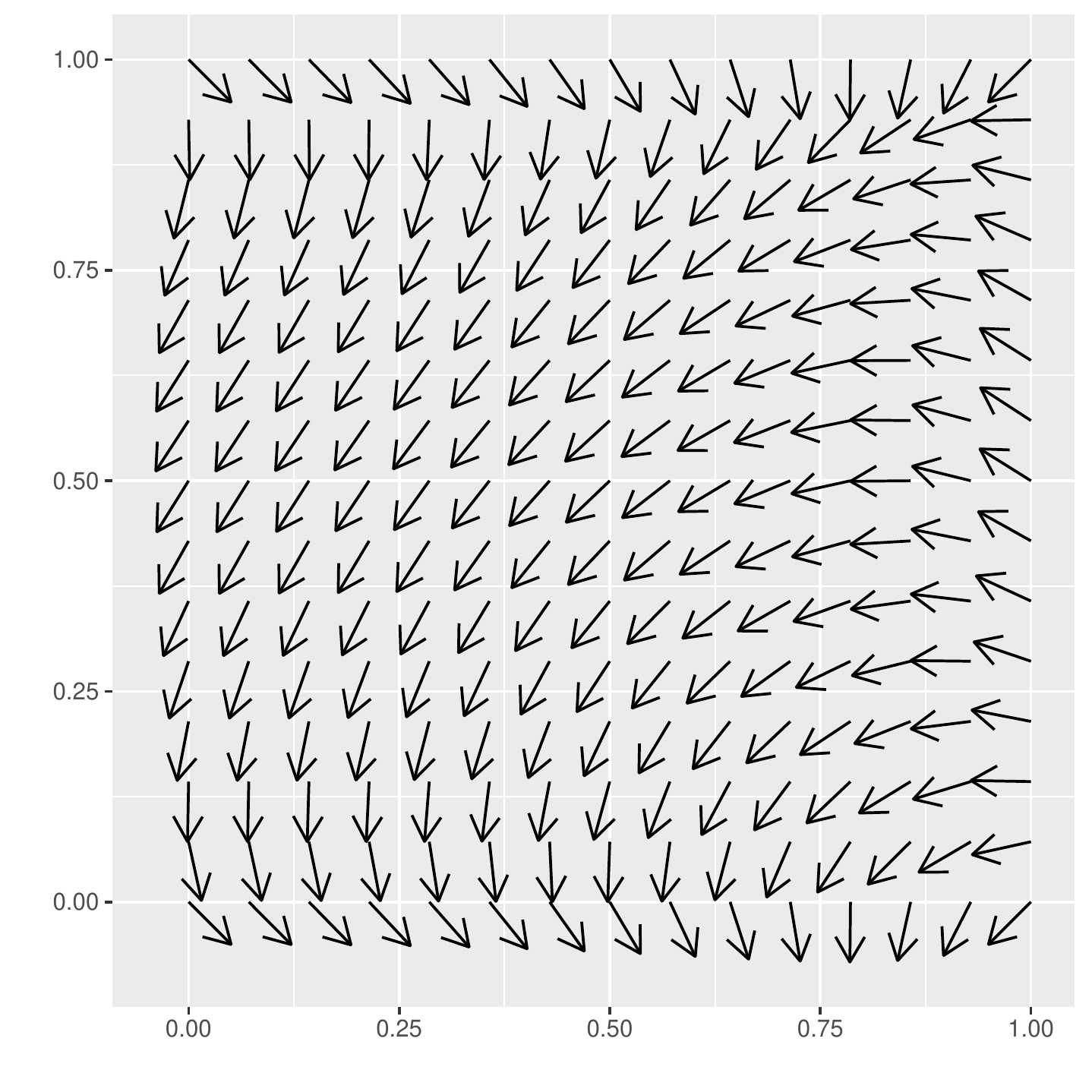}
	\includegraphics[width=0.32\textwidth]{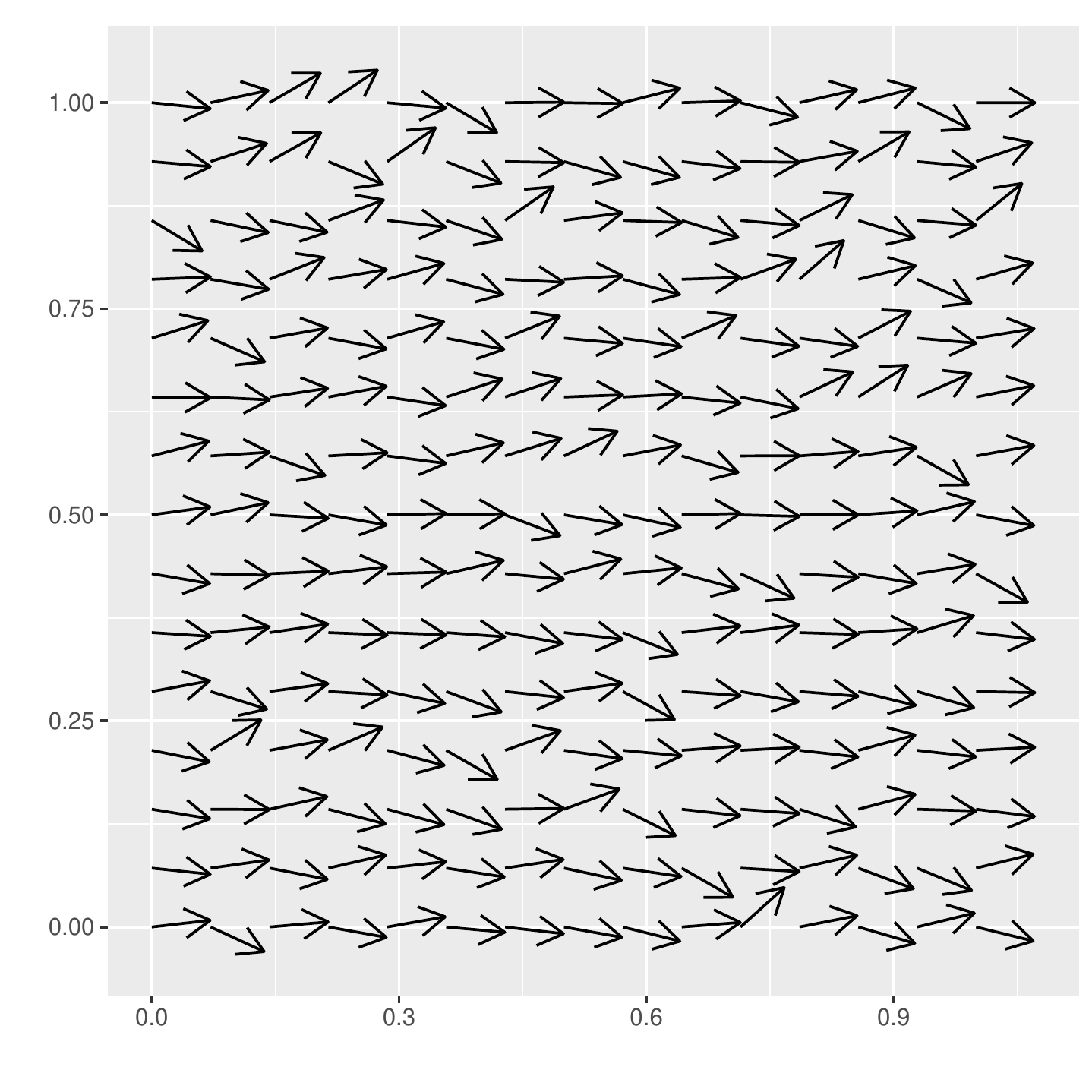}
	\includegraphics[width=0.32\textwidth]{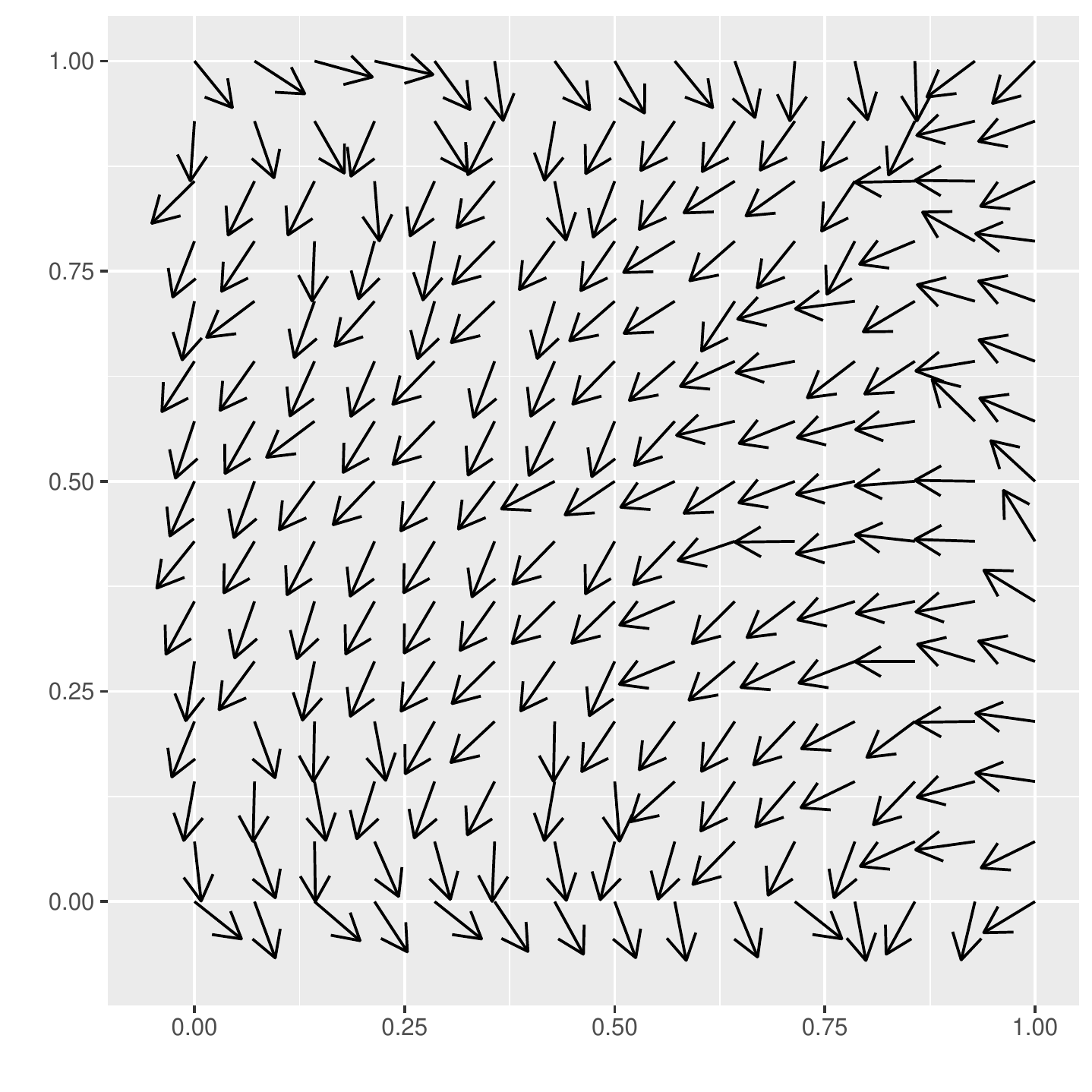}
	\caption{Illustration of model generation (model M1: top row; model M2: bottom row) on a $15\times 15$ grid. In left panels, regression functions evaluated at the grid points. In center panels, independent errors from a von Mises distribution with zero mean and concentration $\kappa=5$, for model M1, and $\kappa=15$, for model M2. In right panels, random response variables  obtained by adding the two previous plots.}
	\label{figure:sim}
\end{figure}

Numerical and graphical outputs summarize the finite sample performance of NW  and LL type  estimators in the different scenarios. In all cases, the smoothing parameter is chosen by cross-validation, selecting the bandwidth matrix $\mathbf{H}$ that minimizes the function:
$${\rm CV}(\mathbf{H})=\sum_{i=1}^n \left\{  1-\cos\left[\Theta_i-\hat{m}_{\mathbf{H}}^{(i)}(\mathbf{X}_i;p)\right]\right\},$$ 
where $\hat{m}_{\mathbf{H}}^{(i)}(\cdot;p)$ stands for the NW type estimator ($p=0$) or the LL type estimator ($p=1$), computed using all observations except $(\mathbf{X}_i,\Theta_i)$.
Taking into account the type of regression functions considered in models M1 and M2 and to speed up the computing times, in this simulation study, the
bandwidth matrix is restricted to be diagonal with possibly different elements. A multivariate Epanechnikov kernel is considered for simulations.

Table \ref{table:simus} shows the average, over the 500 replicates, of the circular average squared error (CASE), defined as:
\begin{eqnarray}\label{CASE}
{\rm CASE}[\hat{m}_{\mathbf{H}}(\mathbf{x};p)]=\frac{1}{n}\sum_{i=1}^n \left\{1- \cos\left[m(\mathbf{X}_i) - \hat{m}_{\mathbf{H}}(\mathbf{X}_i;p)\right]\right\},\end{eqnarray}
for models M1 and M2, and $p=0$ (NW) and $p=1$ (LL). It can be seen that this average error decreases with the sample size, and it is smaller for the LL type estimator in all the considered scenarios.

\begin{table}[h!]
	\centering
	\begin{tabular}{ccccccccccccc}
		%		\hline
		$\kappa$&$n$&	&	 &\text{M1}&&&&& \text{M2}& \\
		\cline{4-6}\cline{9-11}
		&&&NW&&LL&&&NW& &LL \\
		\hline
		5&64&&0.0225&& 0.0233 &&&0.0366 &&0.0282\\
		&100&& 0.0170&&0.0164 &&& 0.0387 &&0.0208\\
		&225&& 0.0057 && 0.0049&&&0.0184  &&0.0102\\ 
		&400&&  0.0055 && 0.0049 &&& 0.0128 &&0.0073\\  \hline
		10&64&&0.0120&& 0.0124 &&&0.0212 &&0.0143\\
		&100&& 0.0107&&0.0085 &&& 0.0023 &&0.0012\\
		&225&& 0.0048 && 0.0038&&&0.0124  &&0.0060\\ 
		&400&&  0.0034 && 0.0025 &&& 0.0079 &&0.0042\\  \hline
		15&64&&0.0088&& 0.0088 &&&0.0164 &&0.0106\\
		&100&& 0.0079&&0.0060 &&& 0.0151 &&0.0082\\
		&225&& 0.0037 && 0.0028&&&0.0107  &&0.0046\\ 
		&400&&  0.0025 && 0.0017 &&& 0.0061 &&0.0032\\  \hline
	\end{tabular}
	\caption{Average error (over $500$ replicates) of the CASE given in (\ref{CASE}), for regression models M1  and M2, using NW  and LL type  estimators. Errors are generated from a von Mises distribution with different concentration parameters ($\kappa=5, 10, 15$). Bandwidth matrix is selected by cross-validation.}\label{table:simus}
\end{table}

Numerical outputs are completed with some additional plots. As an illustration of the correct performance of NW and LL type estimators, Figure \ref{figure:regression} shows the theoretical regression functions for models M1 and M2 (left panels) and the corresponding average, over 500 replicates, of the estimates, using the specific scenarios considered in Figure \ref{figure:sim} (NW and LL estimates in the center and right panels, respectively). Notice that, for comparison purposes, the theoretical regression functions are plotted  in a $100\times 100$ regular grid of the explanatory variables (the same grid where the estimations were computed). Plots in the top row present the results for the data genera\-ted from model M1 and those in the bottom row for model M2. Although both estimators have a similar and correct  behavior, the LL estimator seems to show a slightly better performance, at least, for these samples. More reliable comparisons between NW  and LL type  estimators can be performed computing the circular bias (CB), the circular variance (CVAR), and the circular mean squared error (CMSE) for both estimators, in a grid of values of the explanatory variables. These quantities, at a point $\mathbf{x}$, are defined as:
\begin{equation}
{\rm CB}[\hat{m}_{\mathbf{H}}(\mathbf{x};p)]={\rm E}\{\sin[\hat{m}_{\mathbf{H}}(\mathbf{x};p) - m(\mathbf{x})]\},
\label{cb}
\end{equation}
\begin{equation}
{\rm CVAR}[\hat{m}_{\mathbf{H}}(\mathbf{x};p)]={\rm E}\{1-\cos[\hat{m}_{\mathbf{H}}(\mathbf{x};p) - \mu(\mathbf{x;p})]\},
\label{cvar}
\end{equation}
\begin{equation}
{\rm CMSE}[\hat{m}_{\mathbf{H}}(\mathbf{x};p)]={\rm E}\{1-\cos[m(\mathbf{x})-\hat{m}_{\mathbf{H}}(\mathbf{x};p)]\},
\label{cmse}
\end{equation}
where $\mu(\mathbf{x;p})$ in CVAR denotes the circular mean of $\hat{m}_{\mathbf{H}}(\mathbf{x};p)$. Notice that, using Taylor expansions,  equations (\ref{cb}), (\ref{cvar}) and (\ref{cmse}) are equivalent to the Euclidean versions of these expressions \citep{kim2017multivariate}.

Figures \ref{figure:bias_var_CMSE_m1} and \ref{figure:bias_var_CMSE_m2} show, in the scenarios considered in Figure \ref{figure:sim}, the CB, CVAR and CMSE computed in a $100\times 100$ regular grid of the explanatory variables, when using NW (top row) and LL (bottom row) fits, for models M1 and M2, respectively. The expectations in (\ref{cb}), (\ref{cvar}) and (\ref{cmse}) are approxi\-mated by the averages over the 500 replicates generated. It can be seen that the NW type estimator ($p=0$) provides larger biases and smaller variances than the LL type estimator ($p=1$) in both settings. However, the CMSE is smaller for the LL fit in most of the grid points. Similar results for the CB, CVAR and CMSE for both estimators were obtained in other scenarios.

\begin{figure}[t]
	\centering
	\includegraphics[width=1\textwidth]{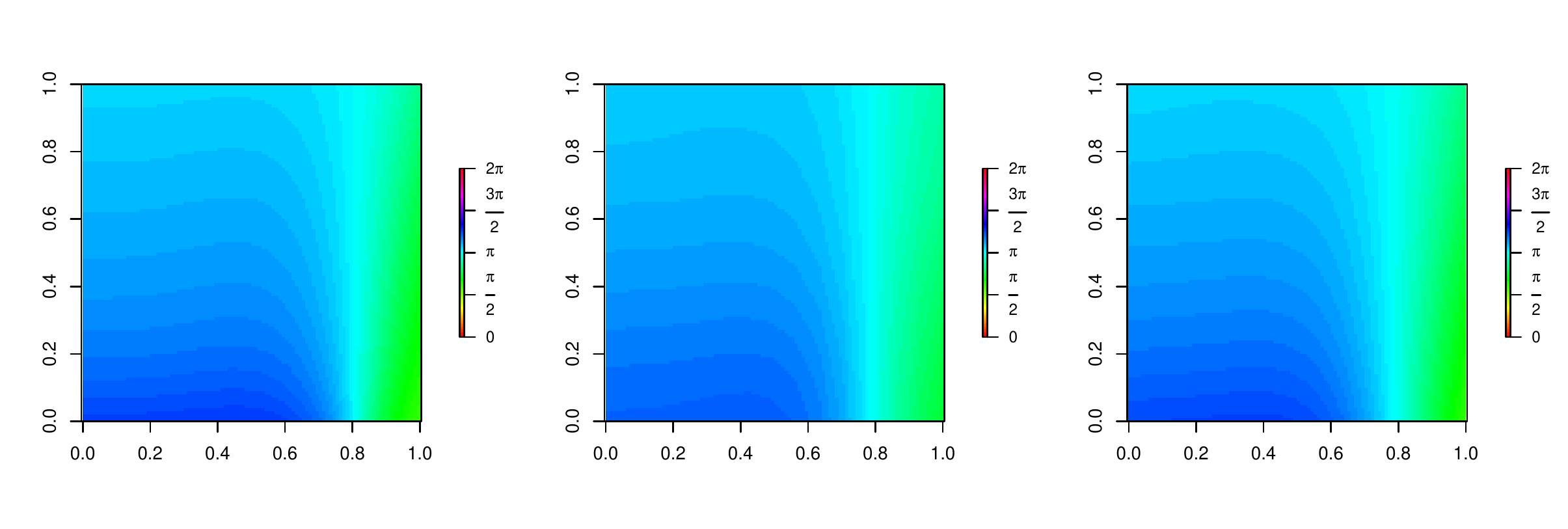}
	\includegraphics[width=1\textwidth]{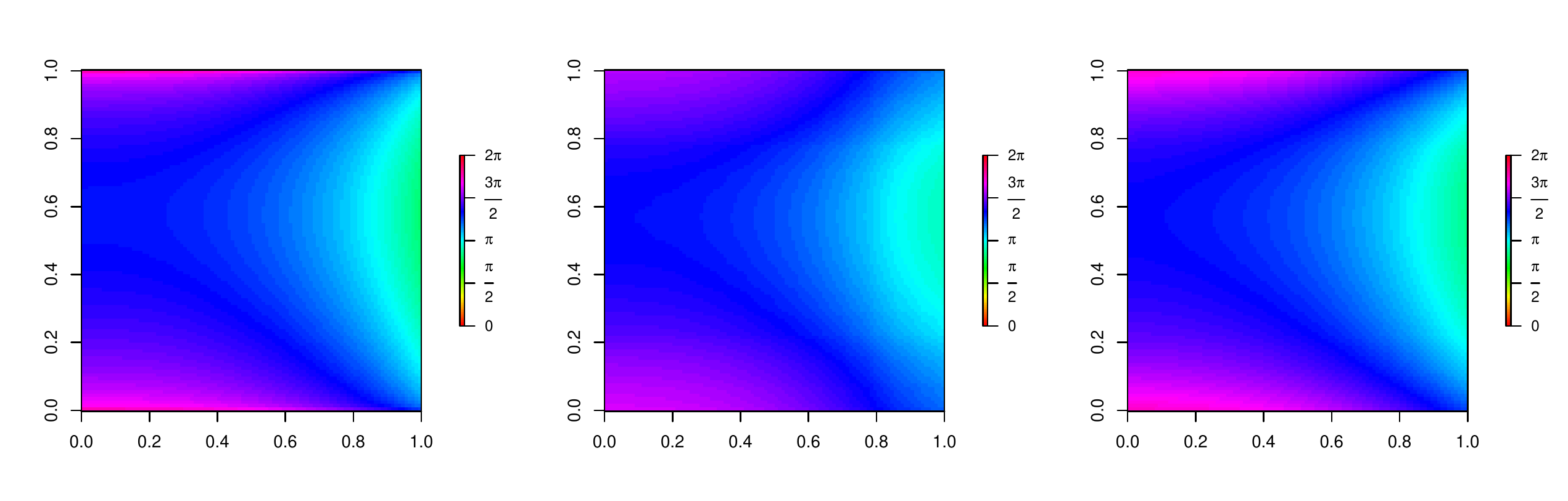}
	\caption{Theoretical regression function (left), jointly with the  average, over 500 replicates, of NW (center) and LL (right)  estimates, using the specific scenarios considered in Figure \ref{figure:sim}, for model M1 (top row) and model M2 (bottom row).}
	\label{figure:regression}
\end{figure}

\begin{figure}[t]
	\centering
	\includegraphics[width=1\textwidth]{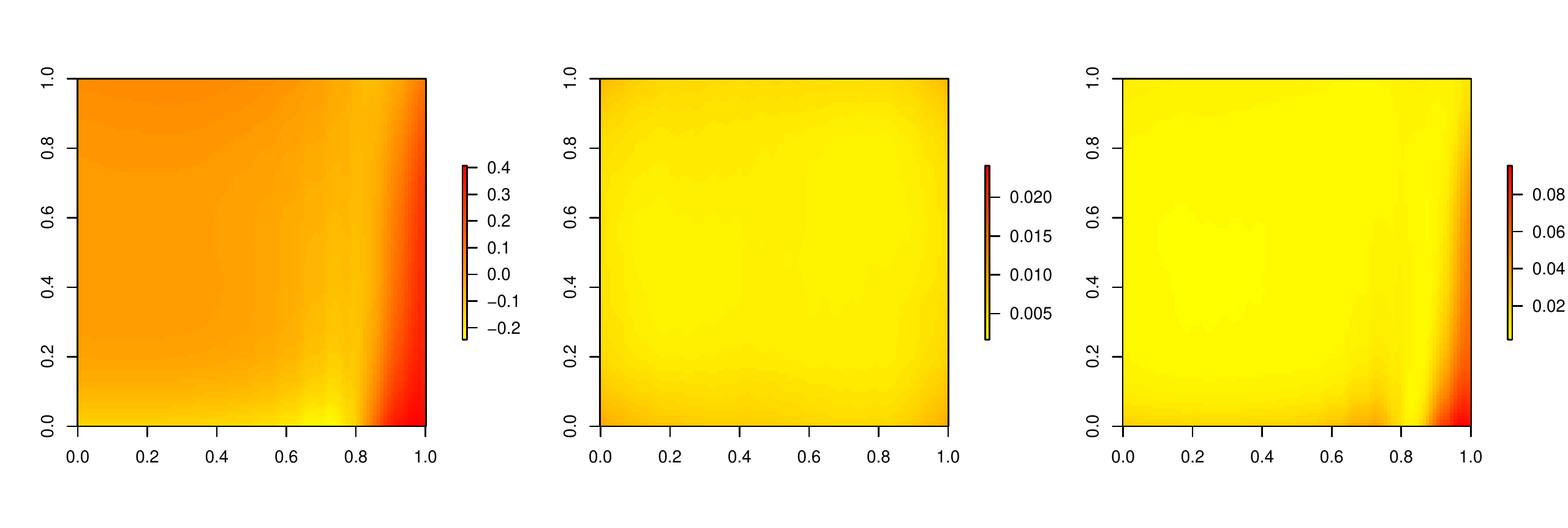}
	\includegraphics[width=1\textwidth]{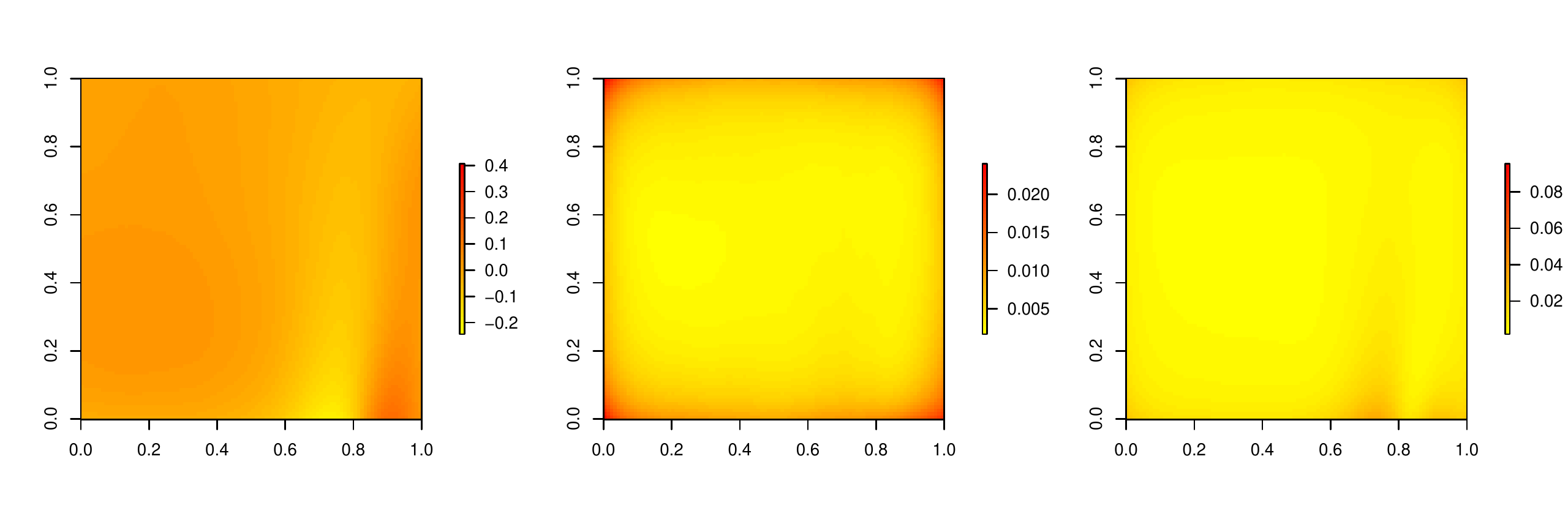}
	\caption{Circular bias (left), circular variance (center) and CMSE (right) surfaces for model M1 for a $100\times 100$ regular grid, using NW (top row) and LL (bottom row) fits. $n=225$ and von Mises errors with zero mean and $\kappa=5$.}
	\label{figure:bias_var_CMSE_m1}
\end{figure}

\begin{figure}[t]
	\centering
	\includegraphics[width=1\textwidth]{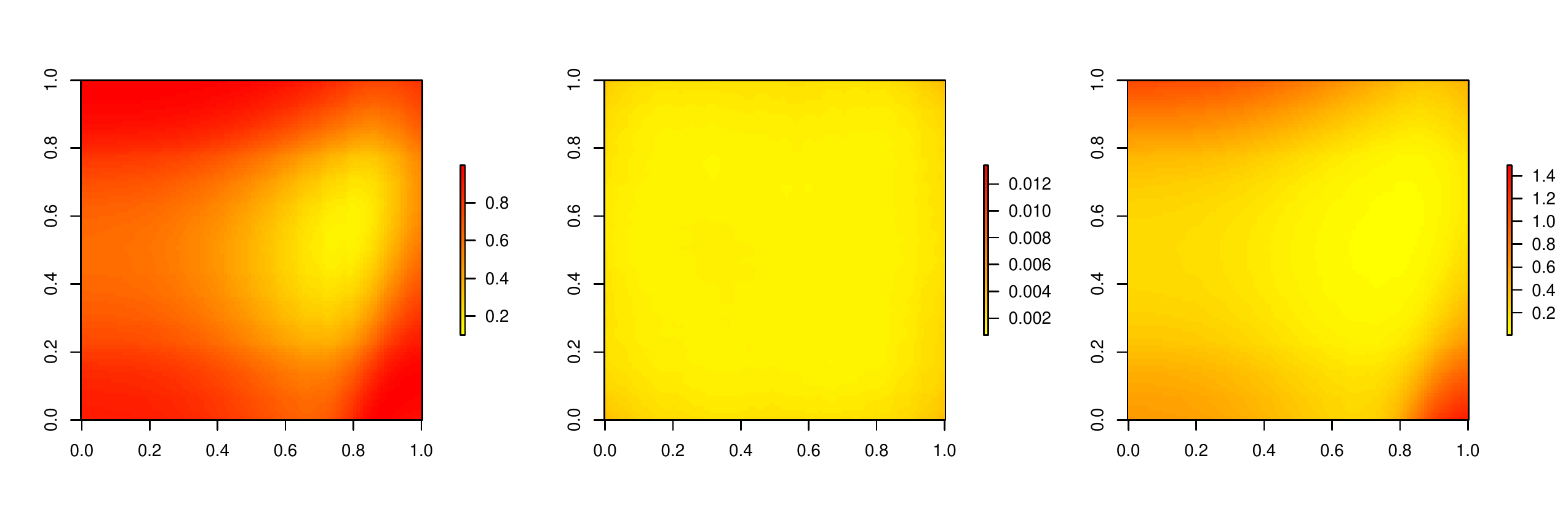}
	\includegraphics[width=1\textwidth]{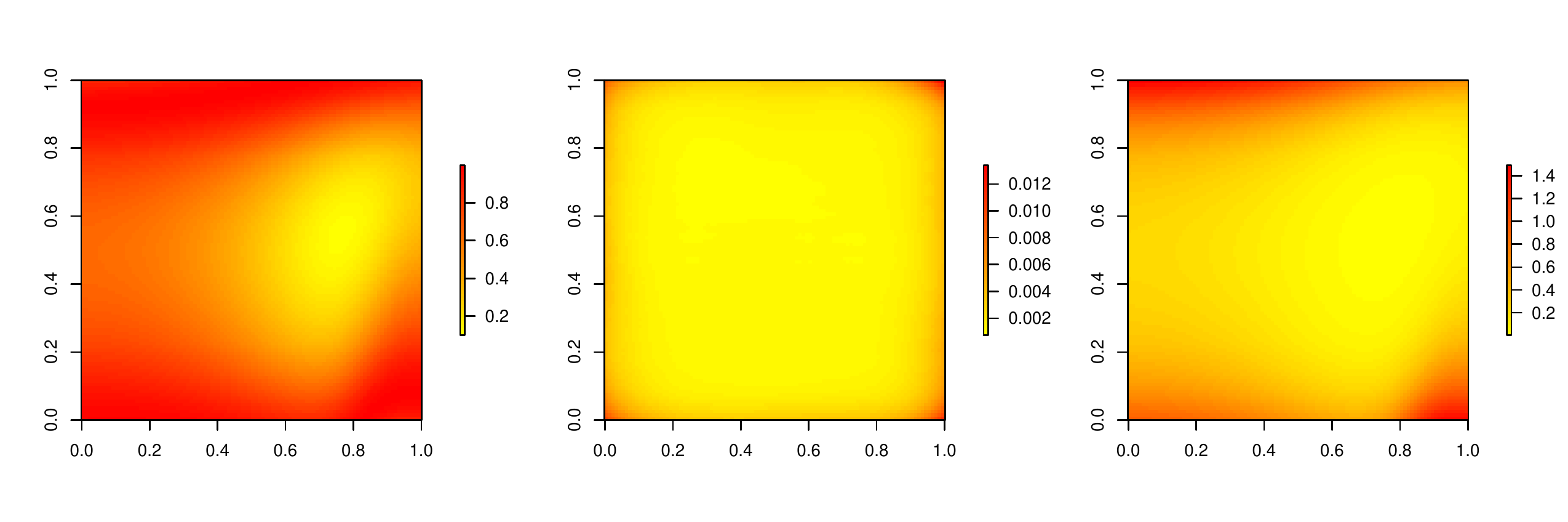}
	\caption{Circular bias (left), circular variance (center) and CMSE (right) surfaces for model M2 for a $100\times 100$ regular grid, using NW (top row) and LL (bottom row) fits. $n=225$ and von Mises errors with zero mean and $\kappa=15$.}
	\label{figure:bias_var_CMSE_m2}
\end{figure}

\section{Real data example}
\label{sec:example}
A real data example is presented in order to illustrate the application of the proposed estimators. Based on the simulation study, where the LL type estimator presented a slightly better performance than the NW one, just results corres\-ponding to $\hat{m}_{\mathbf{H}}(\mathbf{x};1)$ are provided for real data.  The orientation of two species of sand hoppers, considering parametric multiple regression methods for circular responses, following the proposal in \cite{presnell1998projected}, were analyzed in \cite{scapini2002multiple}. This is a parametric approach that assumes a projected normal distribution for the scape directions and the corresponding parameters (circular mean and mean resultant vector) depend on the explanatory variables through a linear model. We refer to \cite{scapini2002multiple} and \cite{marchetti2003use} for details on the experiment, a thorough data analysis and sound biological conclusions. Dealing with the same data set, in \cite{marchetti2003use}, the authors conclude that the orientation is different for the two sexes (males and females) and they explicitly mention that nonparametric smoothers are flexible tools that may suggest unexpected features of the data. So, the illustration with our proposal is a first attempt to analyze this data set with nonparametric tools in order to check how orientation (in degrees) behaves when temperature (in Celsius degrees) and (relative) humidity (in percentage) are included as covaria\-tes. For illustration purposes, only observations corresponding to (relative) humidity values larger than 45\% are considered in this analysis. The corres\-ponding data sets are plotted in  Figure \ref{figure:data} (males in the left panel and females in the right panel), being the sample sizes $n=330$ and $n=404$, for male and female sand hoppers, respectively.
\begin{figure}[hbt]
	\centering
	\includegraphics[width=0.55\textwidth]{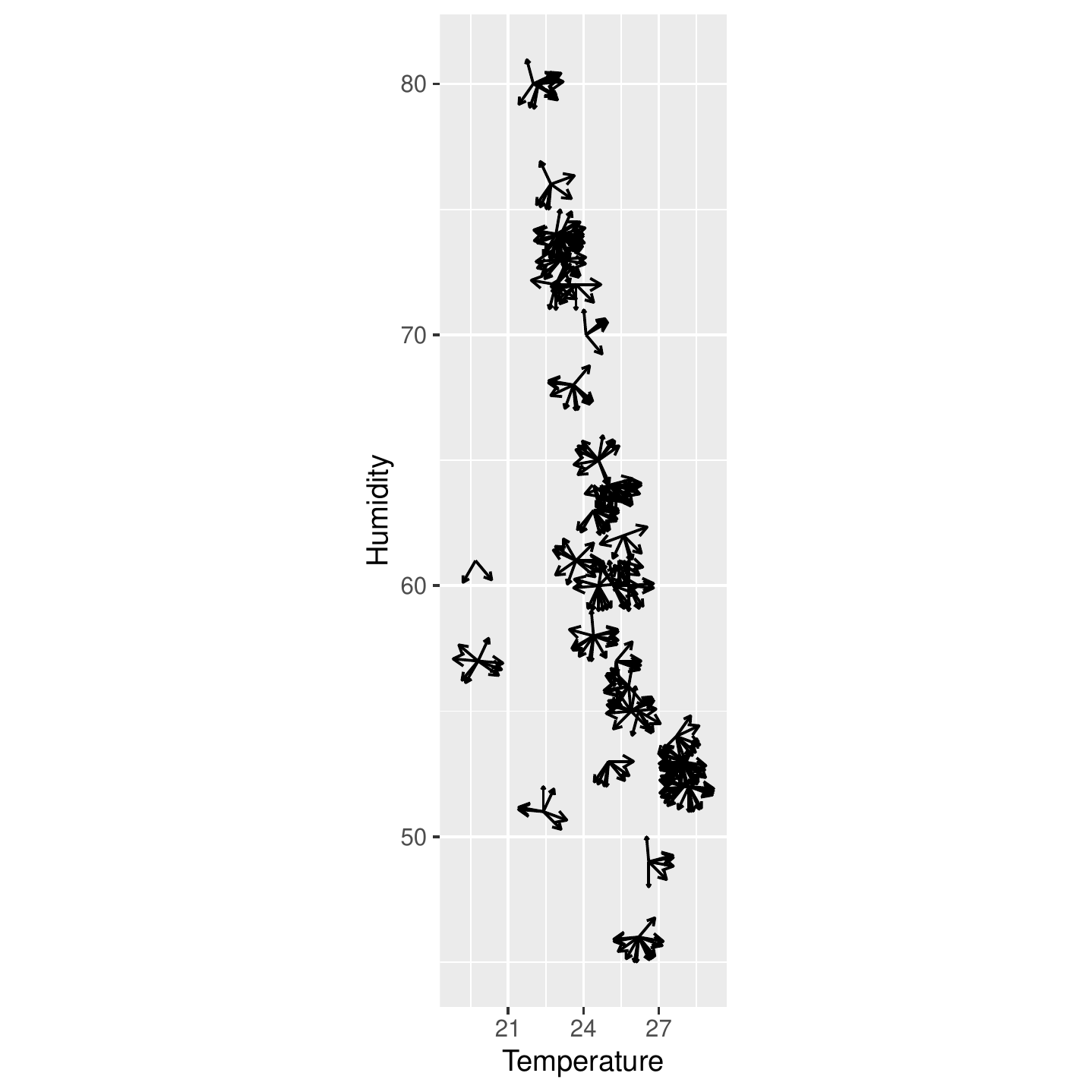}\hspace{-2cm}
	\includegraphics[width=0.55\textwidth]{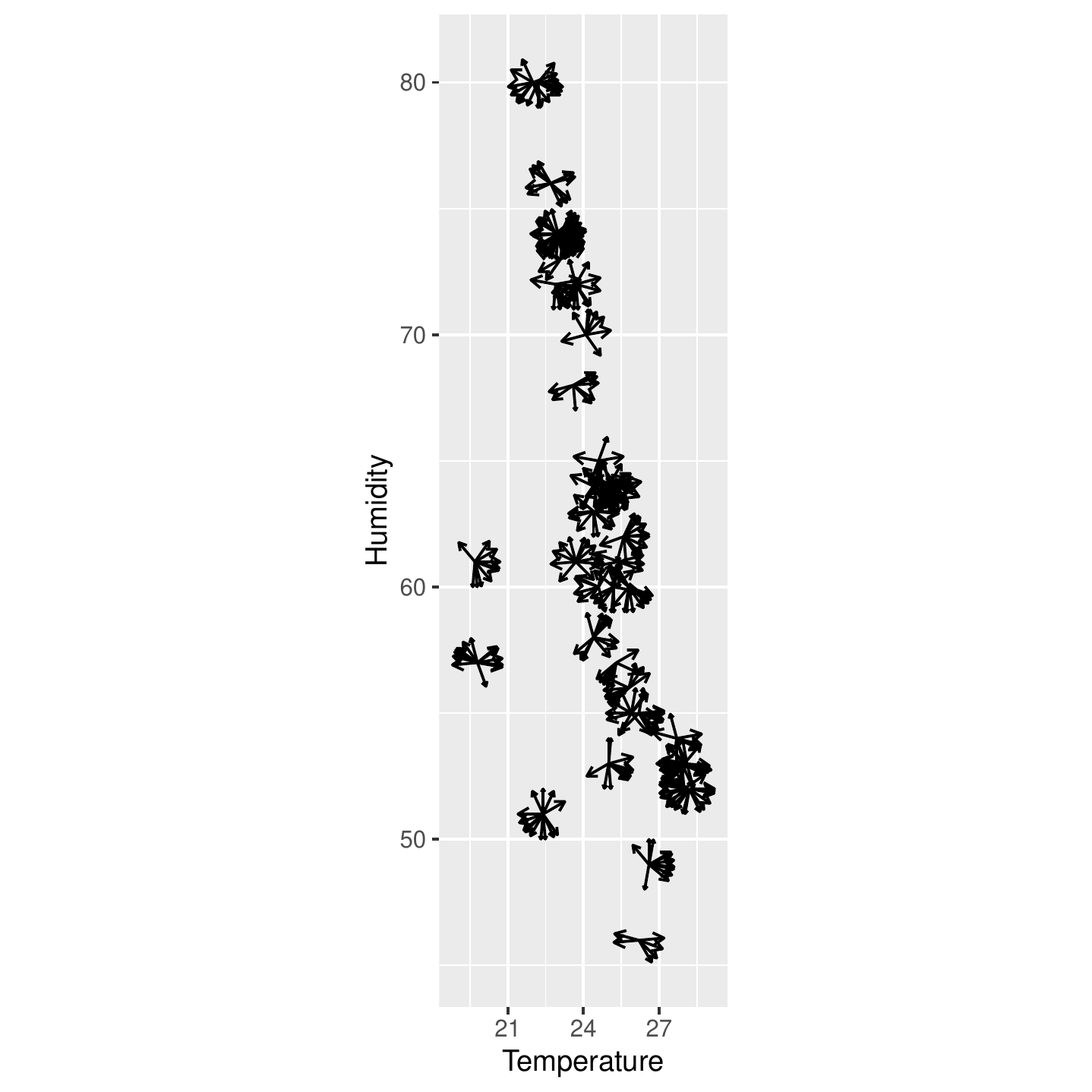}
	\caption{Observed orientation of male (left) and female (right) sand hoppers as a function of temperature and relative humidity.}
	\label{figure:data}
\end{figure}

Figure \ref{figure:data_est} shows the LL estimates for male (left) and female (right) mean orientations, considering temperature (horizontal axis) and relative humidi\-ty (vertical axis) as covariates.  Note that measurements of temperature and humidity are the same for males and females, given that these values corres\-pond to experimental conditions. In this example, unlike in the simulation experiments, the CV bandwidth matrix has been searched in the family of the symmetric and definite positive full bandwidth matrices, using an optimization algorithm based on the Nelder--Mead simplex method described in \cite{lagarias1998convergence}. Using the initial bandwidth matrix $\mathbf{H}_{init}=1.5 \cdot \mbox{diag}\left\{\hat{\sigma}_{X_1} ,\hat{\sigma}_{X_2}  \right\}$, the algorithm converged to
$$
\mathbf{H}_{\rm CV}^{m}=\left[ \begin{array}{cc}
2.7781 & 0.0001 \\ 
0.0001 & 15.2529
\end{array} \right],
$$
for males, and to
$$
\mathbf{H}_{\rm CV}^{f}=\left[\begin{array}{cc}
4.0930 & -0.0009 \\ 
-0.0009 & 13.1937
\end{array} \right],
$$
for females, where $\hat{\sigma}_{X_1}$ and $\hat{\sigma}_{X_2}$ denote the sample standard deviations of the covariates $X_1=$``tem\-pe\-ra\-tu\-re'' and $X_2=$``humidity'', respectively. As in the previous section, a multivariate Epanechnikov kernel is considered.  Note that the estimation grid of explanatory variables on which the estimates of the mean were computed was constructed by overlying the survey values of temperature and humidity with a $100 \times 100$ grid and, then, dropping every grid point that did
not satisfy one of the following two requirements: (a) it is within 15 ``grid cell length''  from an
observation point, or (b) the calculation for the estimates of the sine and cosine components at that grid point uses a smoothing vector that is sufficiently stable.
Both requirements are admittedly somewhat arbitrary, but they represent
a compromise between coverage over the region of interest and ability to avoid singular design
matrices. Even with these restrictions, some of the estimates for low temperature values (around 20 Celsius degrees) seem to be spurious, specially in the case of male individuals. This can be due to data sparseness or a bounda\-ry effect, two well-known situations where kernel-based smoothing methods may present certain drawbacks. Trying to avoid some of these problems and taking into account that there are repeated values of the covariates, possibly due to rounded measurements, additional estimates have been obtained after jittering the original data (the corresponding plots are not shown), obtaining estimates that follow similar patterns to those shown in Figure \ref{figure:data_est}. The \emph{mean direction} followed by male and female sand hoppers is different for some  temperature and humidity conditions. Seawards orientation was roughly $7\pi/4$, so it can be seen that females are more seawards oriented than males, specially for mid to low values of temperature. 
\begin{figure}[t]
	\centering
	\includegraphics[width=0.49\textwidth]{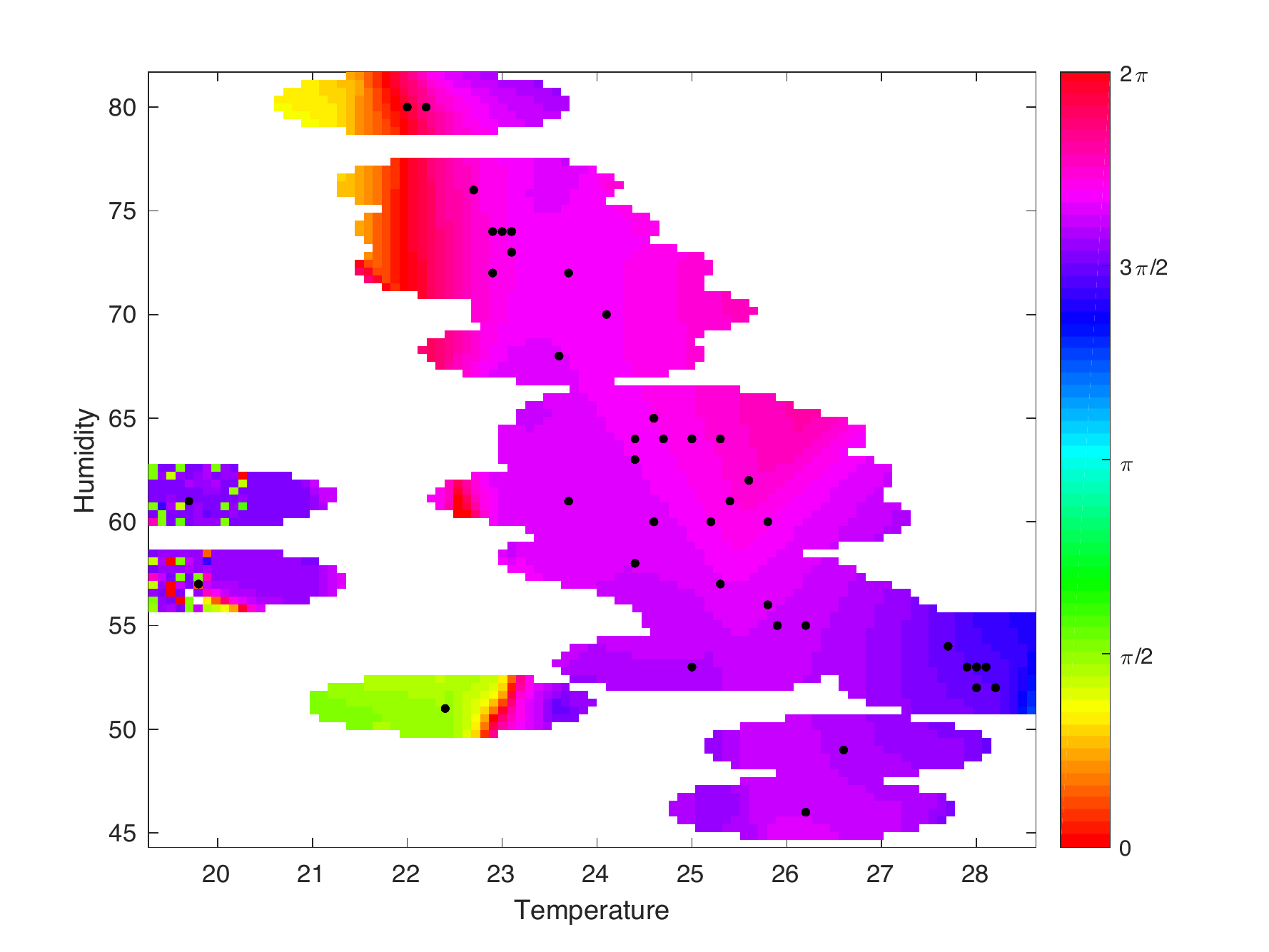}
	\includegraphics[width=0.49\textwidth]{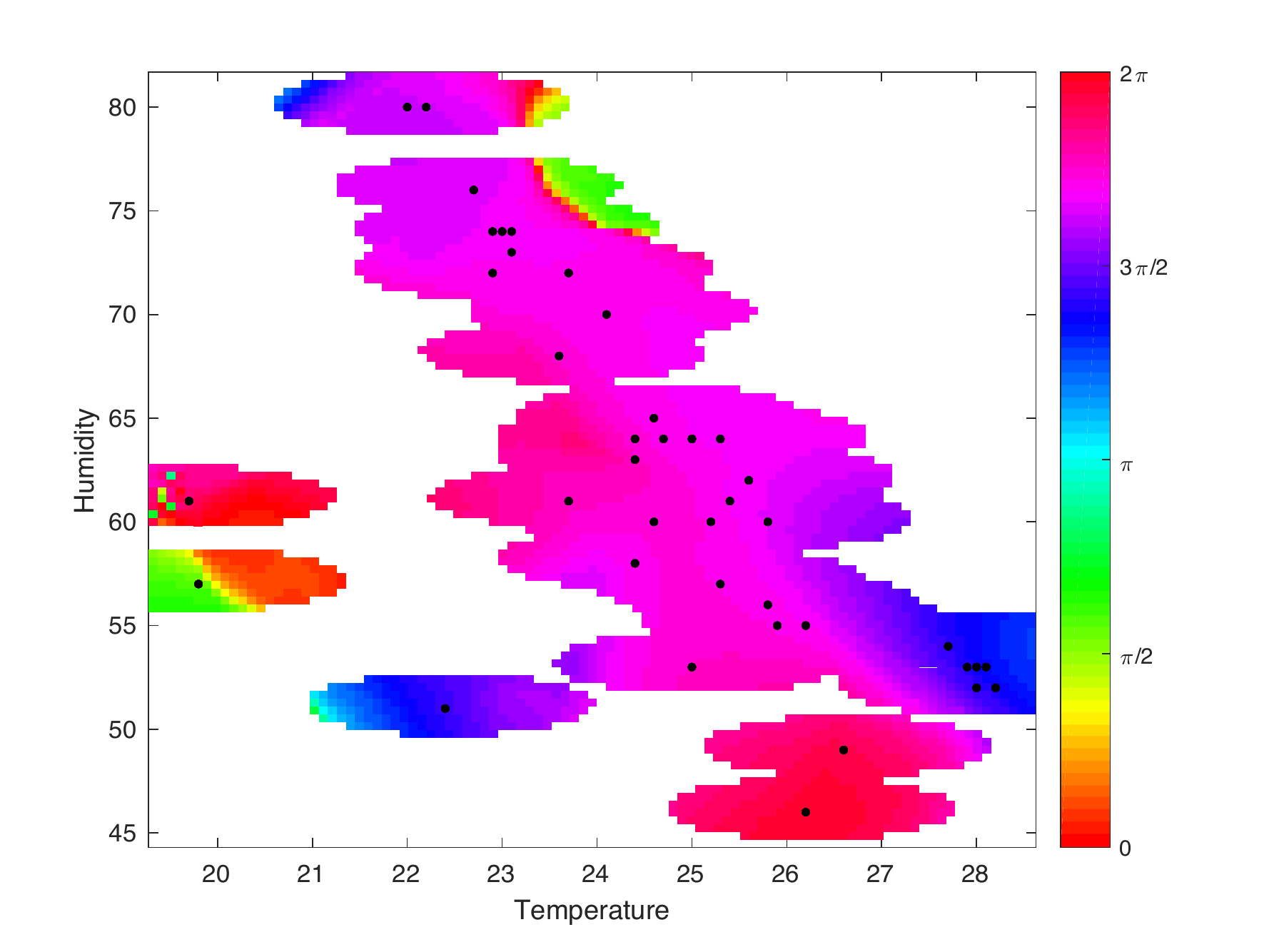}
	\caption{Estimates of the mean orientation of males (left) and females (right) sand hoppers, considering a LL estimator with cross-validation bandwidth. Horizontal axis: temperature, in Celsius degrees. Vertical axis: relative humidity, in percentage.}
	\label{figure:data_est}
\end{figure}
%-----------------------------------%

%---------------------------------------------------%
\section*{Discussion}\label{discussion}
Nonparametric regression estimation for circular responses and $\mathbb R^d$-valued covariates is studied in this paper. Our proposal considers kernel-based approa\-ches, with special attention on NW  and LL type  estimators in general dimension, and for higher order polynomials in the one-dimensional case. Asymptotic conditional bias and variance are derived and the performance of the estimators is assessed in a simulation study.

For practical implementation, the selection of a $d$-dimensional bandwidth matrix is required. 
In the regression Euclidean context, the bandwidth selection problem has been widely addressed in the last decades \citep[see, for example][where a review on bandwidth selection methods for kernel regression is provided]{kohler2014review}. More related to the topic of the present paper,  a rule-of-thumb and a bandwidth rule for selection scalar or diagonal bandwidth matrices for multivariate local linear regression with real-valued response and $\mathbb R^d$-valued covariate is derived in \cite{yang1999multivariate}. Also in this context, in \cite{ana_wences}, a bootstrap method to estimate the mean squared error and the smoothing parameter for the multidimensional regression local linear estimator is proposed. However, in the framework of nonparametric regression methods for circular variables, the research on bandwidth selection is very scarce or non-existent. Our practical results are derived with a cross-validation bandwidth given that, up to our knowledge, there are no other bandwidth selectors available in this context. The design of alternative procedures to select the bandwidth matrix for the estimators studied in this paper based, for example, on bootstrap methods are indeed of great interest. This problem is out of the scope of the present paper, but it is an interesting topic of research for a future study.

Once the problem of including a $\mathbb R^d$-valued covariate for explaining the behaviour of a circular response is solved, it seems natural to think about the consideration of covariates of different nature. Since the proposed estimator is constructed by considering the atan2 of the smooth estimators of the regression functions for the sine and cosine components of the response, an adaptation of our proposal for different types of covariates implies the use of suitable weights. For instance, if a  spherical (circular, as a particular case) or a mixture of spheri\-cal and real-valued covariates are considered to influence a circular response, weights for estimating the sine and cosine components could be constructed following the ideas in \cite{garcia2013kernel} for cylindrical density estimation. If a categorical covariate is included in the model, a similar approach to the one in \cite{racine2004nonparametric} or in  \cite{li_racine} could be also followed. In all these cases, bandwidth matrices should be selected, and cross-validation techniques could be applied.

The results obtained in Theorem \ref{C_teo2} and \ref{C_teo3} can be extended to an arbitrary dimension $d$ of the space of the covariates by using the asymptotic properties for  $\hat m_{j,\mathbf H}(\mathbf x;p)$ provided in \cite{gu2015multivariate} who considered the leading term of biases and varian\-ces of multivariate  local polynomial estimators of general order $p$.  Results on the asymptotic distribution of multivariate local polynomial estimators are also provided in \cite{gu2015multivariate}. The joint asymptotic normality of $\hat m_{1,\mathbf H}(\mathbf x;p)$ and $\hat m_{2,\mathbf H}(\mathbf x;p)$ can be used to derive, via the delta-method, the asymptotic distribution of statistics which can be expressed in terms of $\hat m_{1,\mathbf{H}}(\mathbf{x}; p)$ and $\hat m_{2,\mathbf{H}}(\mathbf{x}; p)$. For example, a suitable adaptation of Proposition 3.1 of \cite{jammalamadaka2001topics} can be used to derive the limiting distribution of the tangent of $\hat m_{\mathbf{H}}(\mathbf{x};p)$. 

In our scenario, data generated from the regression model are assumed to be independent. However, in many practical situations, this assumption does not seem reasonable (e.g. data area collected over time or space). The simple construction scheme behind the proposed class of estimators makes possible to easily obtain asymptotic properties in more general frameworks. As an example, when data are not i.i.d. but are realizations of stationary processes satisfying some mixing conditions, the results provided in \cite{masry1996multivariate} can be used. It should be also noted that, when the data exhibit some kind of dependence, although the expression for the estimator will be the same, this structure will affect the estimator variance and should be taking into account to select properly the bandwidth parameter, as in \cite{francisco2005smoothing}.

%---------------------------------------------------%
\section*{Acknowledgements}\label{acknowledgements}
	The authors acknowledge the support from the Xunta de Galicia grant ED481A-2017/361 and the European Union (European Social Fund - ESF). This research has been partially supported by MINECO grants  MTM2016-76969-P and MTM2017-82724-R, and by the Xunta de Galicia (Grupo de Referencia Competitiva  ED431C-2017-38, and Centro de Investigación del SUG ED431G 2019/01), all of them through the ERDF.  The authors thank Prof. Felicita Scapini and his research team 
	who kindly provided the sand hoppers data that are used in this work.  Data were collected within the Project ERB ICI8-CT98-0270 from the European Commission, Directorate General XII Science.

\bibliographystyle{apa-good}
	\bibliography{bibib3_v2}

\section*{Appendix. Proof of the results}
This section is devoted to present the proofs of Theorem \ref{teoNW}, \ref{teoLL}, \ref{C_teo2} and \ref{C_teo3}. More specifically, the asymptotic properties of the proposed nonparametric regression estimator $\hat{m}_{\mathbf{H}}(\mathbf{x};p)$, for $p=0,1$, are established in Theorem \ref{teoNW} and \ref{teoLL}, respectively. For $d=1$, the extensions for $p=2$ and $p=3$ are considered in Theorem \ref{C_teo2} and \ref{C_teo3}, respectively. 

\begin{proof}[Proof of Theorem \ref{teoNW}]
	First, to obtain the bias of $\hat{m}_{\mathbf{H}}(\mathbf{x};0)$, using the same linea\-rization arguments as in the proof of Theorem $1$ of \cite{di2013non},  $\mbox{atan2}(\hat{m}_{1, \mathbf{H}},\hat{m}_{2, \mathbf{H}})$ is expanded in Taylor series around $(m_1,m_2)$, where for simplicity,  $\hat m_{j,\mathbf{H}}$ and $m_j$  denote  $\hat m_{j,\mathbf{H}}(\mathbf{x})$ and $m_j(\mathbf{x})$, respectively, for $j=1,2$, to get
	\begin{eqnarray}\label{atan}
\mbox{atan2}(\hat{m}_{1, \mathbf{H}},\hat{m}_{2, \mathbf{H}})\nonumber&=&\mbox{atan2}(m_1,m_2)+\dfrac{m_2}{m_1^2+m_2^2}(\hat{m}_{1, \mathbf{H}}-m_1)\nonumber\\&&- \dfrac{m_1}{m_1^2+m_2^2}(\hat{m}_{2, \mathbf{H}}-m_2)+\dfrac{m_1m_2}{(m_1^2+m_2^2)^2}(\hat{m}_{2, \mathbf{H}}-m_2)^2\nonumber\\&&- \dfrac{m_1m_2}{(m_1^2+m_2^2)^2}(\hat{m}_{1, \mathbf{H}}-m_1)^2\nonumber\\&&- \dfrac{m_1^2-m_2^2}{(m_1^2+m_2^2)^2}(\hat{m}_{1, \mathbf{H}}-m_1)(\hat{m}_{2, \mathbf{H}}-m_2)\nonumber\\&&+\mathcal{O}\left[(\hat{m}_{1, \mathbf{H}}-m_1)^3\right]+\mathcal{O}\left[(\hat{m}_{2, \mathbf{H}}-m_2)^3\right],\end{eqnarray}

	Taking expectations, noting that ${\rm E}\left[(\hat{m}_{j,\mathbf{H}} -m_j)^2\mid\mathbf{X}_1,\dots,\mathbf{X}_n\right]={\rm Var}(\hat{m}_{j, \mathbf{H}}\mid\mathbf{X}_1,\dots,\mathbf{X}_n)+[E(\hat{m}_{j, \mathbf{H}})-m_j\mid\mathbf{X}_1,\dots,\mathbf{X}_n]^2$, and using the results in Proposition \ref{pro1}, it is obtained that 
	\begin{eqnarray*}		&&{\rm E}[\hat{m}_{\mathbf{H}}(\mathbf{x};0)-m(\mathbf{x})\mid\mathbf{X}_1,\dots,\mathbf{X}_n]\\&=&\dfrac{1}{2}{\dfrac{m_2(\mathbf{x})}{m_1^2(\mathbf{x})+m_2^2(\mathbf{x})}\mu_2(K)\tr\left[\mathbf{H}^2\bm{\mathcal{H}}_{m_1}(\mathbf{x})\right]}\\&&+  \dfrac{m_2(\mathbf{x})}{m_1^2(\mathbf{x})+m_2^2(\mathbf{x})}\dfrac{\mu_2(K)}{f(\mathbf{x})}\bm{\nabla}^T{m_1}(\mathbf{x})\mathbf{H}^2\bm{\nabla}{f}(\mathbf{x})\\&&- \dfrac{1}{2}{\dfrac{m_1(\mathbf{x})}{m_1^2(\mathbf{x})+m_2^2(\mathbf{x})}\mu_2(K)\tr\left[\mathbf{H}^2\bm{\mathcal{H}}_{m_2}(\mathbf{x})\right]}\\&&- \dfrac{m_1(\mathbf{x})}{m_1^2(\mathbf{x})+m_2^2(\mathbf{x})}\dfrac{\mu_2(K)}{f(\mathbf{x})}\bm{\nabla}^T{m_2}(\mathbf{x})\mathbf{H}^2\bm{\nabla}{f}(\mathbf{x})\\&&+  {\mathpzc{o}_{\mathbb{P}}[\tr(\mathbf{H}^2)]}.\end{eqnarray*}
	
	Therefore,
	\begin{eqnarray*}			\lefteqn{{\rm E}[\hat{m}_{\mathbf{H}}(\mathbf{x};0)-m(\mathbf{x})\mid\mathbf{X}_1,\dots,\mathbf{X}_n]}\\&=&\dfrac{1}{2}\dfrac{\mu_2(K)}{m_1^2(\mathbf{x})+m_2^2(\mathbf{x})}\tr\bigg\{\mathbf{H}^2\bigg[m_2(\mathbf{x})\bm{\mathcal{H}}_{m_1}(\mathbf{x})-m_1(\mathbf{x})\bm{\mathcal{H}}_{m_2}(\mathbf{x})\bigg]\bigg\}\\&&+  \dfrac{\mu_2(K)}{[m_1^2(\mathbf{x})+m_2^2(\mathbf{x})]f(\mathbf{x})}\bigg\{\left[m_2(\mathbf{x})\bm{\nabla}^T{m_1}(\mathbf{x})-m_1(\mathbf{x})\bm{\nabla}^T{m_2}(\mathbf{x})\right]\mathbf{H}^2\bm{\nabla}{f}(\mathbf{x})\bigg\}\\&&+  {\mathpzc{o}_{\mathbb{P}}[\tr(\mathbf{H}^2)]}.\end{eqnarray*}
	
		Now, taking into account that 
	\begin{eqnarray}
	{\bm{\nabla}} m(\mathbf{x})&=&\dfrac{1}{\ell^2(x)}\left[{\bm{\nabla}} m_1(\mathbf{x}) m_2(\mathbf{x})-{\bm{\nabla}} m_2(\mathbf{x}) m_1(\mathbf{x})\right],\label{grad}\\
	{\bm{\mathcal{H}}}_{m}(\mathbf{x})&=&	\dfrac{1}{\ell^2(x)}\left[{\bm{\mathcal{H}}}_{m_1}(\mathbf{x})m_2(\mathbf{x})+{\bm{\nabla}} m_1(\mathbf{x}){\bm{\nabla}}^T m_2(\mathbf{x})\right.\nonumber\\&&- \left. {\bm{\nabla}} m_2(\mathbf{x}){\bm{\nabla}}^T m_1(\mathbf{x})-{\bm{\mathcal{H}}}_{m_2}(\mathbf{x})m_1(\mathbf{x})\right]-\dfrac{2}{\ell(x)}{\bm{\nabla}} \ell(\mathbf{x}){\bm{\nabla}}^Tm(x),\label{hess}
	\end{eqnarray}		
	it follows that
	\begin{eqnarray*}
		{\mathbb{E}}[\hat{m}_{\mathbf{H}}(\mathbf{x};0)-m(\mathbf{x})\mid\mathbf{X}_1,\dots,\mathbf{X}_n]&=&\dfrac{1}{2}\mu_2(K){\rm tr}[\mathbf{H}^2{\bm{\mathcal{H}}}_{m}(\mathbf{x})]+\dfrac{\mu_2(K)}{\ell(\mathbf{x})f(\mathbf{x})}{\bm{\nabla}}^Tm(\mathbf{x})\mathbf{H}^2{\bm{\nabla}}  (\ell f)(\mathbf{x})\\&&+\mathpzc{o}_{\mathbb{P}}[{\rm tr}(\mathbf{H}^2)].
	\end{eqnarray*}
	
	To derive the variance, the function $\mbox{atan2}^2(\hat{m}_{1, \mathbf{H}},\hat{m}_{2, \mathbf{H}})$ is expanded  in Taylor series around $(m_1,m_2)$, to obtain
\begin{eqnarray}\label{atan2}	\mbox{atan2}^2(\hat{m}_{1, \mathbf{H}},\hat{m}_{2, \mathbf{H}})\nonumber&=&\mbox{atan2}^2({m}_1,{m}_2)+\dfrac{2\mbox{atan2}(m_1,m_2)m_2}{m_1^2+m_2^2}(\hat{m}_{1, \mathbf{H}}-{m}_1)\nonumber\\&&- \dfrac{2\mbox{atan2}(m_1,m_2)m_1}{m_1^2+m_2^2}(\hat{m}_{2, \mathbf{H}}-{m}_2)\nonumber\\&&+ \dfrac{2\mbox{atan2}(m_1,m_2)m_1m_2}{(m_1^2+m_2^2)^2}(\hat{m}_{2, \mathbf{H}}-{m}_2)^2\nonumber\\&&- \dfrac{2\mbox{atan2}(m_1,m_2)m_1m_2}{(m_1^2+m_2^2)^2}(\hat{m}_{1, \mathbf{H}}-{m}_1)^2\nonumber\\&&- \dfrac{2\mbox{atan}(m_1,m_2)(m_1^2-m_2^2)}{(m_1^2+m_2^2)^2}(\hat{m}_{1, \mathbf{H}}-{m}_1)(\hat{m}_{2, \mathbf{H}}-{m}_2)\nonumber\\&&+ \dfrac{m_1^2}{(m_1^2+m_2^2)^2}(\hat{m}_{2, \mathbf{H}}-{m}_2)^2+\dfrac{m_2^2}{(m_1^2+m_2^2)^2}(\hat{m}_{1, \mathbf{H}}-{m}_1)^2\nonumber\\&&- \dfrac{2m_1m_2}{(m_1^2+m_2^2)^2}(\hat{m}_{1, \mathbf{H}}-{m}_1)(\hat{m}_{2, \mathbf{H}}-{m}_2)\nonumber\\&&+\mathcal{O}\left[(\hat{m}_{1, \mathbf{H}}-{m}_1)^3\right]+\mathcal{O}\left[(\hat{m}_{2, \mathbf{H}}-{m}_2)^3\right].\end{eqnarray}
	
	So,  noting that ${\mathbb{V}{\rm ar}}(\hat{m}_{ \mathbf{H}}\mid\mathbf{X}_1,\dots,\mathbf{X}_n)={\mathbb{E}}\left[(\hat{m}_{\mathbf{H}} )^2\mid\mathbf{X}_1,\dots,\mathbf{X}_n\right]-[\mathbb{E}(\hat{m}_{ \mathbf{H}})\mid\mathbf{X}_1,\dots,\mathbf{X}_n]^2$ and taking expectations in the Taylor expansions (\ref{atan}) and (\ref{atan2}), it can be obtained that the conditional variance is:
	\begin{eqnarray*}	\lefteqn{{\rm Var}[\hat{m}_{\mathbf{H}}(\mathbf{x};0)\mid\mathbf{X}_1,\dots,\mathbf{X}_n]}\\&=&\dfrac{m_1^2(\mathbf{x})}{\left[m_1^2(\mathbf{x})+m_2^2(\mathbf{x})\right]^2}{\rm Var}[\hat{m}_{2, \mathbf{H}}(\mathbf{x};0)\mid\mathbf{X}_1,\dots,\mathbf{X}_n]\\&&+  \dfrac{m_2^2(\mathbf{x})}{\left[m_1^2(\mathbf{x})+m_2^2(\mathbf{x})\right]^2}{\rm Var}[\hat{m}_{1, \mathbf{H}}(\mathbf{x};0)\mid\mathbf{X}_1,\dots,\mathbf{X}_n]\\&&- \dfrac{2m_1(\mathbf{x})m_2(\mathbf{x})}{\left[m_1^2(\mathbf{x})+m_2^2(\mathbf{x})\right]^2}\mbox{Cov}[\hat{m}_{1, \mathbf{H}}(\mathbf{x};0),\hat{m}_{2, \mathbf{H}}(\mathbf{x};0)\mid\mathbf{X}_1,\dots,\mathbf{X}_n]\\&&+  \mathcal{O}\left[(\hat{m}_{1, \mathbf{H}}(\mathbf{x};0)-{m}_1(\mathbf{x}))^3\right]+\mathcal{O}\left[(\hat{m}_{2, \mathbf{H}}(\mathbf{x};0)-{m}_2(\mathbf{x}))^3\right].\end{eqnarray*}
	
	Regarding the conditional covariance between $\hat{m}_{1, \mathbf{H}}(\mathbf{x};0)$ and $\hat{m}_{2, \mathbf{H}}(\mathbf{x};0)$, it follows that	
	\begin{eqnarray}
		\lefteqn{\mbox{Cov}[\hat{m}_{1, \mathbf{H}}(\mathbf{x};0),\hat{m}_{2, \mathbf{H}}(\mathbf{x};0)\mid\mathbf{X}_1,\dots,\mathbf{X}_n]}\nonumber\\&=&\dfrac{\sum_{i=1}^{n}\sum_{j=1}^{n}K_\mathbf{H}(\mathbf{X}_i-\mathbf{x})K_\mathbf{H}(\mathbf{X}_j-\mathbf{x})}{\sum_{i=1}^{n}K_\mathbf{H}(\mathbf{X}_i-\mathbf{x})\sum_{j=1}^{n}K_\mathbf{H}(\mathbf{X}_j-\mathbf{x})}\mbox{Cov}[\sin(\Theta_i),\cos(\Theta_j)\mid\mathbf{X}_1,\dots,\mathbf{X}_n]\nonumber\\&=&\dfrac{\sum_{i=1}^{n}K^2_\mathbf{H}(\mathbf{X}_i-\mathbf{x})c(\mathbf{X}_i)}{\left[\sum_{i=1}^{n}K_\mathbf{H}(\mathbf{X}_i-\mathbf{x})\right]^2}\nonumber\\ &=&\dfrac{R(K)c(\mathbf{x})}{n\abs{\mathbf{H}}f(\mathbf{x})}+\mathpzc{o}_{\mathbb{P}}\left(\dfrac{1}{n\abs{\mathbf{H}}}\right).
	\label{cov_m1_m2}
	\end{eqnarray}
	
	Therefore, using (\ref{cov_m1_m2}) and Proposition \ref{pro1}, one gets that
	\begin{eqnarray*}{\rm Var}[\hat{m}_{\mathbf{H}}(\mathbf{x};0)\mid\mathbf{X}_1,\dots,\mathbf{X}_n]&=&\frac{1}{n \abs{\mathbf{H}}}R(K)\dfrac{m_1^2(\mathbf{x})s_2^2(\mathbf{x})}{\left[m_1^2(\mathbf{x})+m_2^2(\mathbf{x})\right]^2f(\mathbf{x})}\\&&+  \frac{1}{n \abs{\mathbf{H}}}R(K)\dfrac{m_2^2(\mathbf{x})s_1^2(\mathbf{x})}{\left[m_1^2(\mathbf{x})+m_2^2(\mathbf{x})\right]^2f(\mathbf{x})}\\&&- \dfrac{2}{n\abs{\mathbf{H}}}R(K)\dfrac{m_1(\mathbf{x})m_2(\mathbf{x})c(\mathbf{x})}{\left[m_1^2(\mathbf{x})+m_2^2(\mathbf{x})\right]^2f(\mathbf{x})}\\&&+  \mathpzc{o}_{\mathbb{P}}\left(\dfrac{1}{n\abs{\mathbf{H}}}\right).\end{eqnarray*}
	
	Taking into account that $m_1(\textbf{x})=f_1(\textbf{x})\ell(\textbf{x})$ and $m_2(\textbf{x})=f_2(\textbf{x})\ell(\textbf{x})$ and using (\ref{eq:relation_DM}), it is obtained that
		\begin{eqnarray*}{\rm Var}[\hat{m}_{\mathbf{H}}(\mathbf{x};0)\mid\mathbf{X}_1,\dots,\mathbf{X}_n]&=&\dfrac{R(K)\sigma^2_1(\mathbf{x})}{n\abs{\mathbf{H}}\ell^2(\mathbf{x})f(\mathbf{x})}+\mathpzc{o}_{\mathbb{P}}\left(\dfrac{1}{n\abs{\mathbf{H}}}\right).\end{eqnarray*}
\end{proof}

%-----------------------------------------------%
\begin{proof}[Proof of Theorem \ref{teoLL}]
	To obtain the bias of $\hat{m}_{\mathbf{H}}(\mathbf{x};1)$, following the arguments used in the proof of Theorem \ref{teoNW} and using results in Proposition \ref{pro2}, one gets that
	\begin{eqnarray*}	\lefteqn{	{\rm E}[\hat{m}_{\mathbf{H}}(\mathbf{x};1)-m(\mathbf{x})\mid\mathbf{X}_1,\dots,\mathbf{X}_n]}\\&=&\frac{1}{2}\mu_2(K) \dfrac{m_2(\mathbf{x})}{m_1^2(\mathbf{x})+m_2^2(\mathbf{x})}\tr\left[\mathbf{H}^2 \bm{\mathcal{H}}_{m_1}(\mathbf{x})\right]\\&&- \frac{1}{2}\mu_2(K) \dfrac{m_1(\mathbf{x})}{m_1^2(\mathbf{x})+m_2^2(\mathbf{x})}\tr\left[\mathbf{H}^2 \bm{\mathcal{H}}_{m_2}(\mathbf{x})\right]+{\mathpzc{o}_{\mathbb{P}}[\tr(\mathbf{H}^2)]}\\&=&\dfrac{1}{2}{\dfrac{\mu_2(K)}{m_1^2(\mathbf{x})+m_2^2(\mathbf{x})}\tr\left\{\mathbf{H}^2\left[m_2(\mathbf{x})\bm{\mathcal{H}}_{m_1}(\mathbf{x})-m_1(\mathbf{x})\bm{\mathcal{H}}_{m_2}(\mathbf{x})\right]\right\}}\\&&+  {\mathpzc{o}_{\mathbb{P}}[\tr(\mathbf{H}^2)]}.\end{eqnarray*}
	
	Considering (\ref{grad}) and (\ref{hess}), it can be obtained that  
	\begin{eqnarray*}
	\lefteqn{{\mathbb{E}}[\hat{m}_{\mathbf{H}}(\mathbf{x};1)-m(\mathbf{x})\mid\mathbf{X}_1,\dots,\mathbf{X}_n]}\\&=&\dfrac{1}{2}{\mu_2(K)}{\rm tr}\left\{\mathbf{H}^2\left[{\bm{\mathcal{H}}}_{m}(\mathbf{x})+\dfrac{2}{\ell(x)}{\bm{\nabla}} \ell(\mathbf{x}){\bm{\nabla}}^Tm(x)\right]\right\}+\mathpzc{o}_{\mathbb{P}}[{\rm tr}(\mathbf{H}^2)]\\&=&\dfrac{1}{2}\mu_2(K){\rm tr}[\mathbf{H}^2{\bm{\mathcal{H}}}_{m}(\mathbf{x})]+\dfrac{\mu_2(K)}{\ell(\mathbf{x})}{\bm{\nabla}}^Tm(\mathbf{x})\mathbf{H}^2{\bm{\nabla}}  \ell (\mathbf{x})+\mathpzc{o}_{\mathbb{P}}[{\rm tr}(\mathbf{H}^2)]
\end{eqnarray*} 
	
	As for the variance of $\hat{m}_{\mathbf{H}}(\mathbf{x};1)$, the same arguments as those employed in the proof of Theorem \ref{teoNW} to obtain the variance of $\hat{m}_{\mathbf{H}}(\mathbf{x};0)$ can be used. In this case, the conditional covariance between $\hat{m}_{1, \mathbf{H}}(\mathbf{x};1)$ and $\hat{m}_{2, \mathbf{H}}(\mathbf{x};1)$ is:
	\begin{eqnarray*}
		\lefteqn{	\mbox{Cov}[\hat{m}_{1, \mathbf{H}}(\mathbf{x};1),\hat{m}_{2, \mathbf{H}}(\mathbf{x};1)\mid\mathbf{X}_1,\dots,\mathbf{X}_n]\nonumber}\\&=&\mathbf{e}_1^T(\bm{\mathcal{X}}_{\mathbf{x}}^T\bm{\mathcal{W}}_{\mathbf{x}}\bm{\mathcal{X}}_{\mathbf{x}})^{-1}\bm{\mathcal{X}}_{\mathbf{x}}^T\bm{\mathcal{W}}_{\mathbf{x}}\Sigma \bm{\mathcal{W}}_{\mathbf{x}}\bm{\mathcal{X}}_{\mathbf{x}}(\bm{\mathcal{X}}_{\mathbf{x}}^T\bm{\mathcal{W}}_{\mathbf{x}}\bm{\mathcal{X}}_{\mathbf{x}})^{-1}\mathbf{e}_1,
	\end{eqnarray*}
	where $\Sigma$ is the covariance matrix of $\sin(\Theta)$ and $\cos(\Theta)$, whose $(i,j)$ entry is $\Sigma_{i,j}=\mbox{Cov}[\sin(\Theta_i),\cos(\Theta_j)],$ $i,j=1,\dots,n.$
	
	After some calculations, denoting $\mathbf{1}_d$  and  $\mathbf{1}_{d\times d}$ the $d \times 1$ vector and the $d\times d$ matrix with all entries equal to 1, respectively, it can be obtained that
	\begin{eqnarray*}
			\lefteqn{\left(\dfrac{1}{n}\bm{\mathcal{X}}_{\mathbf{x}}^T\bm{\mathcal{W}}_{\mathbf{x}}\bm{\mathcal{X}}_{\mathbf{x}}\right)^{-1}}\\&=&\left(\begin{array}{ll}
			\frac{1}{n}\sum_{i=1}^{n} K_{\mathbf{H}}(\mathbf{X}_i-\mathbf{x}) & \frac{1}{n}\sum_{i=1}^{n} K_{\mathbf{H}}(\mathbf{X}_i-\mathbf{x})(\mathbf{X}_i-\mathbf{x})^T  \\
			\frac{1}{n}\sum_{i=1}^{n} K_{\mathbf{H}}(\mathbf{X}_i-\mathbf{x})(\mathbf{X}_i-\mathbf{x})  & \frac{1}{n}\sum_{i=1}^{n} K_{\mathbf{H}}(\mathbf{X}_i-\mathbf{x})(\mathbf{X}_i-\mathbf{x}) (\mathbf{X}_i-\mathbf{x}) ^T
		\end{array}\right)^{-1}\\&=&
		\left(
		\begin{array}{ll}
			f^{-1}(\mathbf{x})+\mathpzc{o}_{\mathbb{P}}(1) & -f^{-2}(\mathbf{x})\nabla f(\mathbf{x})^T+\mathpzc{o}_{\mathbb{P}}(\mathbf{1}^T_d) \\
			-f^{-2}(\mathbf{x})\nabla f(\mathbf{x})+\mathpzc{o}_{\mathbb{P}}(\mathbf{1}_d)  & \left[\mu_2(K)f(\mathbf{x})\mathbf{H}^2\right]^{-1}+\mathpzc{o}_{\mathbb{P}}(\mathbf{H}\mathbf{1}_{d\times d}\mathbf{H})	\end{array}
		\right).\end{eqnarray*}	
	
	Moreover, denoting
	\begin{eqnarray*}s_{1,n}(\mathbf{x})&=&		\frac{1}{n^2}\sum_{i=1}^{n} K^2_{\mathbf{H}}(\mathbf{X}_i-\mathbf{x})c(\mathbf{X}_i),\\
		s_{2,n}(\mathbf{x})&=&	\frac{1}{n^2}\sum_{i=1}^{n} K^2_{\mathbf{H}}(\mathbf{X}_i-\mathbf{x})(\mathbf{X}_i-\mathbf{x})c(\mathbf{X}_i) ,\\
		s_{3,n}(\mathbf{x})&=& \frac{1}{n^2}\sum_{i=1}^{n} K^2_{\mathbf{H}}(\mathbf{X}_i-\mathbf{x})(\mathbf{X}_i-\mathbf{x}) (\mathbf{X}_i-\mathbf{x}) ^Tc(\mathbf{X}_i),	
	\end{eqnarray*} 
	it follows that
	\begin{eqnarray*}	\dfrac{1}{n^2}\bm{\mathcal{X}}_{\mathbf{x}}^T\bm{\mathcal{W}}_{\mathbf{x}}\Sigma\bm{\mathcal{W}}_{\mathbf{x}}	\bm{\mathcal{X}}_{\mathbf{x}}&=&\left(\begin{array}{ll}
			s_{1,n}(\mathbf{x}) &  s^T_{2,n}(\mathbf{x}) \\
			s_{2,n}(\mathbf{x}) & s_{3,n}(\mathbf{x})
		\end{array}\right)\\&=&\dfrac{1}{n\abs{\mathbf{H}}}\left(
		\begin{array}{ll}
			c(\mathbf{x})f(\mathbf{x})R(K)+\mathpzc{o}_{\mathbb{P}}(1) & \mathpzc{o}_{\mathbb{P}}(\mathbf{1}^T_d) \\
			\mathpzc{o}_{\mathbb{P}}(\mathbf{1}_d)  & \mathpzc{o}_{\mathbb{P}}(\mathbf{1}_{d\times d}).	\end{array}
		\right),\end{eqnarray*}

	Consequently, by straightforward calculations, one gets
	$$
	\mbox{Cov}[\hat{m}_{1, \mathbf{H}}(\mathbf{x};1),\hat{m}_{2, \mathbf{H}}(\mathbf{x};1)\mid\mathbf{X}_1,\dots,\mathbf{X}_n]\nonumber=\dfrac{R(K)c(\mathbf{x})}{n\abs{\mathbf{H}}f(\mathbf{x})}+\mathpzc{o}_{\mathbb{P}}\left(\dfrac{1}{n\abs{\mathbf{H}}}\right),
	$$
	and the variance of $\hat{m}_{\mathbf{H}}(\mathbf{x};1)$ is:
	$$
	{\rm Var}[\hat{m}_{\mathbf{H}}(\mathbf{x};1)\mid\mathbf{X}_1,\dots,\mathbf{X}_n]=\dfrac{R(K)\sigma^2_1(\mathbf{x})}{n\abs{\mathbf{H}}\ell^2(\mathbf{x})f(\mathbf{x})}+\mathpzc{o}_{\mathbb{P}}\left(\dfrac{1}{n\abs{\mathbf{H}}}\right).
	$$
\end{proof}
\begin{proof}[Proof of Theorem \ref{C_teo2}]
	Using the asymptotic properties of the  local quadratic estimator, close expressions of  ${\mathbb{E}}[\hat m_{j, h}(x;2)\mid X_1,\ldots, X_n]$ and ${{\mathbb{V}{\rm ar}}}[\hat m_{j, h}(x;2)\mid X_1,\ldots, X_n]$, for $j=1,2,$ can be obtained. To derive the bias of $\hat{m}_{h}({x};2)$, following the arguments used in the proof of Theorem \ref{teoNW} and \ref{teoLL}, one gets that
	\begin{eqnarray*}		\lefteqn{{\mathbb{E}}[\hat{m}_{{h}}({x};2)-m({x})\mid{X}_1,\dots,{X}_n]}\\&=&\dfrac{h^4 \mu_4(K_{(2)})f'(x)}{3!f(x)}\dfrac{m_2({x})}{m_1^2({x})+m_2^2({x})}m^{(3)}_1(x)+\dfrac{h^4 \mu_4(K_{(2)})}{4!}\dfrac{m_2({x})}{m_1^2({x})+m_2^2({x})}m^{(4)}_1(x)\\&&- \dfrac{h^4 \mu_4(K_{(2)})f'(x)}{3!f(x)}\dfrac{m_1({x})}{m_1^2({x})+m_2^2({x})}m^{(3)}_2(x)-\dfrac{h^4 \mu_4(K_{(2)})}{4!}\dfrac{m_1({x})}{m_1^2({x})+m_1^2({x})}m^{(4)}_2(x)\\&&+ \mathpzc{o}_{\mathbb{P}}({h}^4).\end{eqnarray*}
	
	Therefore,
	\begin{eqnarray*}		\lefteqn{{\mathbb{E}}[\hat{m}_{{h}}({x};2)-m({x})\mid{X}_1,\dots,{X}_n]}\\&=&\dfrac{h^4 \mu_4(K_{(2)})f'(x)}{3!f(x)\ell^2(x)}[{m_2({x})}m^{(3)}_1(x)-{m_1({x})}m^{(3)}_2(x)]\\&&+ \dfrac{h^4 \mu_4(K_{(2)})}{4!\ell^2(x)}[{m_2({x})}m^{(4)}_1(x)-{m_1({x})}m^{(4)}_2(x)]+\mathpzc{o}_{\mathbb{P}}({h}^4).\end{eqnarray*}
	Now, taking into account that 
	\begin{eqnarray}
	m'({x})&=&\dfrac{1}{\ell^2(x)}\left[m'_1({x}) m_2({x})-m'_2({x}) m_1({x})\right],\label{grad_uni}\\
	{m''}({x})&=&	\dfrac{1}{\ell^2(x)}\left[{m''_1}({x})m_2({x})-{m''_2}({x})m_1({x})\right]-\dfrac{2}{\ell(x)}\ell'(x)m'(x),\label{hess_uni}\\
	{m^{(3)}}({x})&=&	\dfrac{1}{\ell^2(x)}\left[{m^{(3)}_1}({x})m_2({x})-{m^{(3)}_2}({x})m_1({x})+{m''_1}({x})m'_2({x})-{m'_1}({x})m''_2({x})\right]\nonumber\\&&- \dfrac{4}{\ell(x)}\ell'(x)m''(x)-\dfrac{2}{\ell^2(x)}\ell'^2(x)m'(x)-\dfrac{2}{\ell(x)}\ell''(x)m'(x),\label{3der_uni}\\
	{m^{(4)}}({x})&=&	\dfrac{1}{\ell^2(x)}\left[{m^{(4)}_1}({x})m_2({x})-{m^{(4)}_2}({x})m_1({x})+2{m^{(3)}_1}({x})m'_2({x})-2{m'_1}({x})m^{(3)}_2({x})\right]\nonumber\\&&- \dfrac{6}{\ell(x)}\ell'(x)m^{(3)}(x)-\dfrac{2}{\ell(x)}\ell^{(3)}(x)m'(x)-\dfrac{6}{\ell(x)}\ell''(x)m''(x)\nonumber\\&&- \dfrac{6}{\ell(x)^2}\ell'^2(x)m''(x)-\dfrac{6}{\ell(x)^2}\ell'(x)\ell''(x)m'(x),\label{4der_uni}
	\end{eqnarray}		
	it follows that
	\begin{eqnarray*}
		\lefteqn{	{\mathbb{E}}[\hat{m}_{{h}}({x};2)-m({x})\mid{X}_1,\dots,{X}_n]}\\&=&\frac{h^{4}\mu_{4}(K_{(2)})f'(x)}{3!f(x)}m^{(3)}(x)\\&&+ \frac{h^{4}\mu_{4}(K_{(2)})f'(x)}{3!f(x)}\left[\dfrac{2\ell''(x)m'(x)}{\ell(x)}+\dfrac{m_2''(x)m_1'(x)-m_1''(x)m_2'(x)]}{\ell^2(x)}\right]\\&&+ \frac{h^{4}\mu_{4}(K_{(2)})f'(x)}{3!f(x)}\left[\dfrac{4\ell'(x)m''(x)}{\ell(x)}+\dfrac{2\ell'^2(x)m'(x)}{\ell^2(x)}\right]\\&&+ \frac{h^4\mu_{4}(K_{(2)})}{4!}m^{(4)}(x)\\&&+ \frac{h^4\mu_{4}(K_{(2)})}{4!}\left[\dfrac{2\ell^{(3)}(x)m'(x)}{\ell(x)}+\dfrac{2m_2^{(3)}(x)m_1'(x)-2m_1^{(3)}(x)m_2'(x)}{\ell^2(x)}\right]\\&&+ \frac{h^4\mu_{4}(K_{(2)})}{4!}\left[\dfrac{6\ell'(x)m^{(3)}(x)+6\ell''(x)m''(x)}{\ell(x)}+\dfrac{6\ell'^2(x)m''(x)+6\ell'(x)\ell''(x)m'(x)}{\ell^2(x)}\right]\\&&+ \mathpzc{o}_{\mathbb{P}}\left(h^{4}\right)
	\end{eqnarray*}

	As for the variance of $\hat{m}_{{h}}({x};2)$, the same arguments as those employed in the proof of Theorem \ref{teoNW} and \ref{teoLL} can be used. The conditional covariance between both $\hat m_{1, h}(x;2)$ and $\hat m_{2, h}(x;2)$ is  $${\mbox{Cov}}[\hat m_{1, h}( x;2),\hat m_{2,h}( x;2)\mid X_1,\ldots, X_n]=\frac{1}{nhf(x)} R(K_{(2)}) c(x)+\mathpzc{o}_{\mathbb{P}}\left(\frac{1}{n h}\right),$$ and the variance of $\hat{m}_{{h}}({x};2)$ is:
	$${\mbox{Var}}[\hat m_{h}(x;2)\mid X_1,\ldots, X_n]=\frac{1}{nh\ell^2(x)f(x)} R(K_{(2)}) \sigma^2_1(x)+\mathpzc{o}_{\mathbb{P}}\left(\frac{1}{n h}\right).$$
	
\end{proof}

\begin{proof}[Proof of Theorem \ref{C_teo3}]To obtain the conditional bias of $\hat{m}_{{h}}({x};3)$,  using the asymptotic properties of the local cubic estimator,  one gets that
	\begin{eqnarray*}		\lefteqn{{\mathbb{E}}[\hat{m}_{{h}}({x};3)-m({x})\mid{X}_1,\dots,{X}_n]}\\&=&\dfrac{h^4 \mu_4(K_{(2)})}{4!}\dfrac{m_2({x})}{m_1^2({x})+m_2^2({x})}m^{(4)}_1(x)-\dfrac{h^4 \mu_4(K_{(2)})}{4!}\dfrac{m_1({x})}{m_1^2({x})+m_1^2({x})}m^{(4)}_2(x)+\mathpzc{o}_{\mathbb{P}}({h}^2)\\&=&\frac{h^4\mu_{4}(K_{(3)})}{4!}m^{(4)}(x)\\&&+ \frac{h^4\mu_{4}(K_{(2)})}{4!}\left[\dfrac{2\ell^{(3)}(x)m'(x)}{\ell(x)}+\dfrac{2m_2^{(3)}(x)m_1'(x)-2m_1^{(3)}(x)m_2'(x)}{\ell^2(x)}\right]\\&&+ \frac{h^4\mu_{4}(K_{(2)})}{4!}\left[\dfrac{6\ell'(x)m^{(3)}(x)+6\ell''(x)m''(x)}{\ell(x)}+\dfrac{6\ell'^2(x)m''(x)+6\ell'(x)\ell''(x)m'(x)}{\ell^2(x)}\right]\\&&+ \mathpzc{o}_{\mathbb{P}}\left(h^{4}\right).\end{eqnarray*}
	
	Reasoning as in the proof of Theorem \ref{C_teo3}, the conditional variance   of $\hat{m}_{{h}}({x};3)$ can be obtained.
\end{proof}

\end{document}